\documentclass[twocolumn,review]{autart}

\usepackage{graphicx}
\usepackage{amsmath}
\usepackage{amssymb}
\usepackage{algorithm,algorithmic}
\usepackage[center]{subfigure}
\usepackage{euscript}
\usepackage{verbatim}
\usepackage{amssymb}
\usepackage{cite}
\usepackage{color}
\usepackage{cancel}
\usepackage{calligra}
\usepackage{pgf}
\usepackage{tikz}
\usepackage{ctable}
\usetikzlibrary{arrows,automata}
\usetikzlibrary{shapes}
\usetikzlibrary{positioning}
\usepackage{bbm}
\usepackage{multirow}
\usepackage{url}

\newtheorem{definition}{\textbf{Definition}}{\normalfont}{\normalfont}
{\normalfont}{\normalfont}
\newtheorem{remark}{Remark}{\normalfont}{\normalfont}
\newtheorem{theorem}{Theorem}
\newtheorem{assumption}{Assumption}

\newtheorem{lemma}{Lemma}
\newenvironment{proof}{\emph{Proof:}}{\hfill$\square$}

\renewcommand{\figurename}{Figure}
\renewcommand{\tablename}{Table}

\newcommand{\smallmat}[1]{\left[ \begin{smallmatrix}#1 \end{smallmatrix} \right]}

\usepackage{enumitem}
\renewcommand{\theenumi}{\arabic{enumi}}

\renewcommand{\theenumii}{\arabic{enumii}}

\renewcommand{\theenumiii}{\arabic{enumiii}}

\begin{document}

	\begin{frontmatter}
		\runtitle{On the design of regularized explicit predictive controllers from input-output data
		}
		\title{On the design of regularized explicit predictive controllers from input-output data\thanksref{funding}
		}
		
		\thanks[funding]{This project was partially supported by the Italian Ministry of University and Research under the PRIN'17 project \textquotedblleft Data-driven learning of constrained control systems", contract no. 2017J89ARP.}
		
		\author[Polimi]{Valentina Breschi}\ead{valentina.breschi@polimi.it},
		\author[Polimi]{Andrea Sassella}\ead{andrea.sassella@polimi.it},
		\author[Polimi]{Simone Formentin}\ead{simone.formentin@polimi.it}
		
		\address[Polimi]{Dipartimento di Elettronica, Informazione e Bioingegneria, Politecnico di Milano, Piazza L. da Vinci 32, 20133 Milano, Italy.}
		

		\begin{abstract}
		On the wave of recent advances in data-driven predictive control, we present an explicit predictive controller that can be constructed from a batch of input/output data only. The proposed explicit law is build upon a regularized implicit data-driven predictive control problem, so as to guarantee the %existence and 
		uniqueness of the explicit predictive controller. As a side benefit, the use of regularization is shown to improve the capability of the explicit law in coping with noise on the data. The effectiveness of the retrieved explicit law and the repercussions of regularization on noise handling are analyzed on two benchmark simulation case studies, showing the potential of the proposed regularized explicit controller.	
		\end{abstract}
		
	\end{frontmatter}
	
	\section{Introduction}
	One of the main challenges of modern control theory is finding the most effective (and efficient) approaches to benefit from data when designing a controller for an unknown system. Traditionally, learning-based control strategies rely on a two-step procedure, which entails the identification of a model for the unknown system and the design of a model-based controller. %Even though t
	These techniques can profit from established tools for system identification \cite{Ljung1999}, but the resulting models are often not optimized for control, as their goal is to approximate the system dynamics by minimizing some fitting error. Although control-oriented identification approaches have been proposed (see \emph{e.g.,} \cite{FORMENTIN2021,GEVERS2005}), they still do not allow one to avoid the two-stage procedure, with the modeling phase frequently making use of the lion's share of time and resources.\\ 
	With a change in the data-handling paradigm, several techniques have been proposed to design controllers \emph{directly} from data, while bypassing an explicit identification phase. Consolidated techniques for data-driven control, such as Virtual Reference Tuning (VRFT) method \cite{CAMPI2002,formentin2019deterministic,formentin2012non}, Iterative Feedback Tuning (IFT) \cite{Hjalmarsson2002} and Correlation-based Tuning (CbT) \cite{Karimi2004,van2011data}, directly employ data to tune the controller, but they have two major drawbacks. Firstly, they rely on the definition of a reference model embedding the desired closed-loop behavior. In this context, reference model selection thus becomes a rather delicate and time consuming task, with the reference model being the main tuning knob of these approaches \cite{selvi2018towards,van2021direct,breschi2021proper}. Secondly, state-of-the-art direct control techniques are not naturally equipped to cope with saturation and constraints, thus requiring additional layers in the control structure (see, \emph{e.g.,} \cite{Piga2018,Masti2020,Breschi2020,breschi20b,Sassella2021}) to handle them.\\
	More recently, the regained popularity of results from behavioral theory \cite{Willems2005} have lead to the introduction of \textit{alternative} data-based control schemes that rely on a trajectory-based description of the system dynamics. These range from passivity-based \cite{Montenbruck2016,Sharf2020,Martin2021,Tanemura2019,Sharf2021,BERBERICH2020d} and model reference \cite{Breschi2021} controllers, to optimal \cite{DePersis2021,DePersis2019,dorfler2021certaintyequivalence} and predictive \cite{Coulson2019,Sassella2021,Berberich2020b,berberich2021design,bongard2021robust,Coulson2021,Dorfler2021} ones, with the latter built to tackle constraints.
		
	\textbf{Contribution and related works.} In the spirit of \cite{Sassella2021}, we derive the explicit solution for the data-based  predictive control problem introduced in \cite{Berberich2021}. Like in traditional Model Predictive Control (MPC) \cite{Alessio2009}, transitioning from a data-based implicit scheme to a data-driven explicit law entails that the optimal control action can be computed via simple function evaluations, rather than requiring the solution of an optimization problem in real-time. This can be computational advantageous, particularly when the problem at hand is relatively simple and the sampling time rather small. Differently from \cite{Berberich2021}, our shift to the explicit predictive law allows one to obtain the optimal input by constructing  a set of data-based matrices only, with the trajectory-based model of the system being ultimately \emph{transparent} to the final user. \\
	Since the structure of the problem proposed in \cite{Berberich2021} prevents the computation of the explicit solution as it is, we propose to augment the performance-oriented predictive control cost with a \textit{regularization} term, acting on the trajectory-based model of the plant. As such, we denominate the presented controller \emph{Regularized Explicit Data-Driven Predictive Controller} (R-EDDPC). Apart from allowing the explicit solution to be found, we show that the regularization can help coping with noisy data, in line with what is proposed in \cite{Berberich2021} to robustify the implicit scheme. We show that the obtained explicit predictive law is piecewise affine, retrieving a controller that resembles the ones introduced in \cite{Sassella2021,Alessio2009}. Differently from \cite{Sassella2021}, the explicit law obtained in this work depends on \textit{input/output data only}, thus not requiring the state of the controlled system to be fully measurable. Nonetheless, we show that the two explicit laws are equivalent under some design assumptions.
	 
	\textbf{Outline.} The paper is organized as follows. After recalling some preliminaries, the problem of designing explicit predictive controllers from data is stated in Section~\ref{Sec:formulation}. In Section~\ref{Sec:fromIMPtoEXP} we shift from the implicit predictive control problem in \cite{Berberich2021} to its explicit counterpart. The proposed solution is then compared with that of \cite{Sassella2021} in Section~\ref{Sec:comparison}. In Section~\ref{Sec:extensions}, we show how the explicit predictive controller can be extended to handle set points changes and to be further robustified against noise. The performance of the proposed method is then discussed in Section~\ref{Sec:examples} by means of two benchmark case studies. The paper is ended by some concluding remarks.
	
	\textbf{Notation.} Let $\mathbb{N}$ and $\mathbb{R}$  be the set of natural and real number respectively. Let $\mathcal{I}_{N}=\{0,1,\ldots,N-1\}$. Denote with $\mathbb{R}^{n}$ the set of real column vectors of dimension $n$ and with $\mathbb{R}^{n \times m}$ the set of real matrices with $n$ rows and $m$ columns. Given $A \in \mathbb{R}^{n \times m}$, we indicate with $A' \in \mathbb{R}^{m \times n}$ its transpose, while we denote by $[A]_{i}$ the $i$-th row of $A$ and by $[A]_{i:j}$ the subset of rows of $A$, starting from the $i$-th up to the $j$-th row (for $i <j$). When $n=m$, we indicate the inverse of $A$ as $A^{-1}$, while $A^{\dagger}$ denotes its right inverse when $n \neq m$. We denote with $I_{n}$ the identity matrix of dimension $n$, while we do not specify the dimension of zero vectors or matrices. Given $x \in \mathbb{R}^{n}$, we denote the squared 2-norm of this vector as $\|x\|^{2}$, while $\|x\|_{Q}^{2}=x'Qx$. Given $Q \in \mathbb{R}^{n \times n}$, $Q \succeq 0$ and $Q \succ 0$ indicate that the matrix is positive semi-definite and positive definite, respectively. Let $Q_{i} \in \mathbb{R}^{n \times n}$, for $i=1,\ldots,L$. Then, $diag(Q_{1},\ldots,Q_{L})$ denotes the block diagonal matrix composed by $\{Q_{i}\}_{i=1}^{L}$. Given a sequence $\{u_{k}\}_{k=0}^{N-1}$, we denote the associated Hankel matrix as $H_{L}(u)$, \emph{i.e.,}
	\begin{equation}\label{eq:Hankel}
		H_{L}(u)=\begin{bmatrix}
			u_{0} & u_{1} & \cdots & u_{N-L}\\
			u_{1} & u_{2} & \cdots & u_{N-L+1}\\
			\vdots & \vdots & \ddots & \vdots\\
			u_{L-1} & u_{L} & \cdots & u_{N-1}
		\end{bmatrix},
	\end{equation}
	while a window of the sequence is indicated as 
	\begin{equation}\label{eq:window}
		u_{[a,b]}=\begin{bmatrix}
			u_{a}\\
			u_{a+1}\\
			\vdots\\
			u_{b}
		\end{bmatrix},
	\end{equation} 
	with $a<b$.
	
	%%%%%%%%%%%%%%%%%%%%%%%%%%%%%%%%%%%%%%%%%%%%%%%%%%%%%%%%%%%%%
	\section{Problem formulation}\label{Sec:formulation}
	The data-driven predictive control formulation proposed in \cite{Berberich2021} represents the starting point from which we derive R-EDDPC. Therefore, we here recall the results on which the former lays its foundation, starting from a formal %introduction of the
	 definition of persistently exciting sequence.  
	\begin{definition}[Persistence of excitation]\label{def:1}
		Given a signal $\nu_k \in \mathbb{R}^{\eta}$, the sequence~$\{\nu_{k}\}_{k=0}^{N-1}$ is said to be  persistently exciting of order $L$ if $\mathrm{rank}(H_{L}(\nu))=\eta L$. 
	\end{definition}
	Consider now a \emph{linear time-invariant} (LTI) system $\mathcal{P}$, with state $x_k \in \mathbb{R}^{n}$, input $u_{k} \in \mathbb{R}^{m}$ and output $y_k \in \mathbb{R}^{p}$. A trajectory of this system can be formally defined as follows.
	\begin{definition}[System trajectory]
		An input/output sequence $\{u_k,y_k\}_{k=0}^{N-1}$ is a trajectory of the LTI system $\mathcal{P}$ if there exists an initial condition $\bar{x} \in \mathbb{R}^{n}$ and a state sequence $\{x_{k}\}_{k=0}^{N}$ such that
		\begin{align*}
			& x_{k+1}=Ax_{k}+Bu_{k},~~x_{0}=\bar{x},\\
			& y_{k}=Cx_{k}+Du_{k},
		\end{align*}
		for all $k \in \mathcal{I}_{N}$, where $(A,B,C,D)$ is a minimal realization of $\mathcal{P}$.
	\end{definition}	
	Given a noiseless trajectory $\{\tilde{u}_{k},\tilde{y}_{k}\}_{k=0}^{N-1}$ of $\mathcal{P}$, with $\{\tilde{u}_{k}\}_{k=0}^{N-1}$ being persistently exciting of order $L+n$, the following result further holds.
	\begin{theorem}[Trajectory-based representation \cite{Berberich2021}]\label{thm:1}
		Let $\{\tilde{u}_{k},\tilde{y_{k}}\}_{k=0}^{N-1}$ be a trajectory of an LTI system $\mathcal{P}$. Assume that $\{\tilde{u}_{k}\}_{k=0}^{N-1}$ is persistently exciting of order $L+n$. Then $\{\bar{u}_{k},\bar{y}_{k}\}_{k=0}^{L-1}$ is a trajectory of $\mathcal{P}$ if and only if there exists $\alpha \in \mathbb{R}^{N-L+1}$ such that
		\begin{equation}
			\begin{bmatrix}
				H_{L}(\tilde{u})\\
				H_{L}(\tilde{y})
			\end{bmatrix}\alpha=\begin{bmatrix}
				\bar{u}_{[0,L-1]}\\
				\bar{y}_{[0,L-1]}
			\end{bmatrix},
		\end{equation}
		where $\tilde{y}=\{\tilde{y}_{k}\}_{k=0}^{N-1}$. 
	\end{theorem}
	According to this result, a single input/output sequence of $\mathcal{P}$ can be used to retrieve a representation of the system spanning the vector space of all its trajectories, provided that such a sequence is properly generated. 

We can now mathematically formulate the problem tackled in this work. Let $\mathcal{P}$ be an \emph{unknown} LTI system of order $n \in \mathbb{N}$, with $m \in \mathbb{N}$ inputs and $p \in \mathbb{N}$ outputs, which is here supposed to be \emph{controllable} and \emph{observable}\footnote{This assumption is shared with \cite{Berberich2021}.}. Assume that we can carry out experiments on $\mathcal{P}$, so as to collect a sequence of $N \in \mathbb{N}$ input/output pairs, $\mathcal{D}_{N}=\{u_{k}^{d},y_{k}^{d}\}_{k=0}^{N-1}$, and assume that the input data satisfy the following condition.
	\begin{assumption}[Quality of data]\label{assump:1}
		The input sequence $u^{d}=\{u_{k}^{d}\}_{k=0}^{N-1}$ is persistently exciting of order $L+2n$, according to Definition~\ref{def:1}.
	\end{assumption} 
	Our goal is to find an \emph{explicit} data-based solution for the implicit data-driven predictive control (DD-PC) problem introduced in \cite{Berberich2021}, that is defined as follows:
	\begin{subequations}\label{eq:nominalMPC}
		\begin{align}
			&\min_{\alpha,\bar{u},\bar{y}}~~ \sum_{k=0}^{L-1} \ell(\bar{u}_{k},\bar{y}_{k}) \label{eq:cost_fcn}\\
			&~~~\mbox{s.t. } \begin{bmatrix}
				\bar{u}_{[-n,L-1]}\\
				\bar{y}_{[-n,L-1]}
			\end{bmatrix}=\begin{bmatrix}
				H_{L+n}(u^{d})\\
				H_{L+n}(y^{d})
			\end{bmatrix}\alpha,\label{eq:model}\\
			&~~~\quad~~ \begin{bmatrix}
				\bar{u}_{[-n,-1]}\\
				\bar{y}_{[-n,-1]}	
			\end{bmatrix}=\chi_{0}, \label{eq:initial_constraint}\\
			&~~~\quad ~~ \begin{bmatrix}
				\bar{u}_{[L-n,L-1]}\\
				\bar{y}_{[L-n,L-1]}	
			\end{bmatrix}=\chi_{L}, \label{eq:terminal_constraint}\\
			&~~~\quad~~~ \bar{u}_{k} \in \mathbb{U}, ~~ \bar{y}_{k} \in \mathbb{Y}~~\forall k \in \mathcal{I}_{L}. \label{eq:feasibility_condition}
		\end{align}
	\end{subequations}
	Therefore, we aim at \emph{explicitly} finding the optimal input sequence $\bar{u}$, the corresponding outputs $\bar{y}$ and the trajectory-based model $\alpha \in \mathbb{R}^{n_{\alpha}}$, with $n_{\alpha}=N\!-\!L\!-\!n\!+\!1$, so as to minimize the quadratic cost
	\begin{equation*}
		\ell(\bar{u}_{k},\bar{y}_{k})=||\bar{u}_{k}-u^{s}||_{R}^{2}+||\bar{y}_{k}-y^{s}||_{Q}^{2},
	\end{equation*}
	over a prediction horizon $L$, with $Q \succeq 0$ and $R \succ 0$, with $u^{s} \in \mathbb{R}^{m}$ and $y^{s} \in \mathbb{R}^{p}$ verifying the subsequent definition.
	\begin{definition}[Equilibrium \cite{Berberich2021}]
			An input/output pair $(u^{s},y^{s})$ is an equilibrium of the LTI system $\mathcal{P}$ if the sequence $\{\bar{u}_{k},\bar{y}_{k}\}_{k=0}^{n}$, with $\bar{u}_{k}=u^{s}$ and $\bar{y}_{k}=y^{s}$ for all $k \in \mathcal{I}_{n}$ is a trajectory of $\mathcal{G}$. 
	\end{definition}
	Meanwhile, our search for the optimal sequences and model is constrained by $(i)$ the initial condition in \eqref{eq:initial_constraint}, $(ii)$ the terminal constraint in \eqref{eq:terminal_constraint}, and $(iii)$ the value constraints in \eqref{eq:feasibility_condition}. The initial condition is characterized by $\chi_{0} \in \mathbb{R}^{n(m+p)}$, which changes at each time instant $t \in \mathbb{N}$. Specifically, this vector collects the $n$ past input/output pairs resulting from the application of the DD-PC law, \emph{i.e.,} at time $t$ it is defined as
	\begin{equation*}
			\chi_{0}=\begin{bmatrix}
				u_{[t-n,t-1]}\\
				y_{[t-n.t-1]}
			\end{bmatrix}.
		\end{equation*} 
	The terminal constraint is instead shaped by a constant vector $\chi_{L} \in \mathbb{R}^{n(m+p)}$, given by 
		\begin{equation*}
			\chi_{L}=\begin{bmatrix}
				u_{n}^{s}\\
				y_{n}^{s}
			\end{bmatrix},
		\end{equation*}
		where $u_{n}^{s}$ and $y_{n}^{s}$ stack $n$ copies of $u^{s}$ and $y^{s}$, respectively. Lastly, the sets characterizing the value constraints in \eqref{eq:feasibility_condition}, namely $\mathbb{U} \subseteq \mathbb{R}^{m}$ and $\mathbb{Y} \subseteq \mathbb{R}^{p}$, are assumed to be \emph{polytopic}. Note that, the terminal constraint further influences the choice of the prediction horizon $L$, since $L \geq n$ for the problem to be well-posed.
		
		\begin{remark}[Stability and recursive feasibility]\label{remark1}
			As proven in \cite{Berberich2021}, within a noiseless setting the data-based formulation in \eqref{eq:nominalMPC} guarantees recursive feasibility and closed-loop exponential stability of the equilibrium $(u^{s},y^{s})$.
		\end{remark}

	\section{From implicit DD-PC to R-EDDPC}\label{Sec:fromIMPtoEXP}
	The stepping stone for the derivation of the explicit data-driven solution to~\eqref{eq:nominalMPC} lays in its reformulation as an optimization problem where the unique decision variable is $\alpha$. To this end, let us define the following matrices
	\begin{subequations}\label{eq:aid_matrices}
		\begin{align}
			& \mathcal{H}_{\gamma}^{P}\!=\![H_{L+n}(\gamma^{d})]_{1:n}, ~~ \mathcal{H}_{\gamma}^{F}\!=\![H_{L+n}(\gamma^{d})]_{n+1:L+n}, \\
			& \mathcal{H}_{\gamma}^{T}\!=\![H_{L+n}(\gamma^{d})]_{L+1:L+n},~~ \mathcal{H}_{\gamma}^{k}\!=\![H_{L+n}(\gamma^{d})]_{k+1}, 
		\end{align}
	\end{subequations}
	where $\gamma$ is a generic placeholder, to be replaced with either $u$ or $y$. We stress that $\mathcal{H}_{u}^{P}$, $\mathcal{H}_{u}^{F}$ and $\mathcal{H}_{u}^{T}$ are all full row rank matrices, since $u^{d}$ is persistently exciting of order $L+2n$. Accordingly, problem \eqref{eq:nominalMPC} can be manipulated and equivalently recast as:
	\begin{subequations}\label{eq:alpha_MPC}
		\begin{align}
			& \min_{\alpha}~~||\alpha||_{W_{d}}^{2}+2c_{d}'\alpha\\ 
			&~~\mbox{s.t. } \begin{bmatrix}
				\mathcal{H}_{u}^{P}\\
				\mathcal{H}_{y}^{P}
			\end{bmatrix}\alpha=\chi_{0},\\
			&~~\quad~~\begin{bmatrix}
				\mathcal{H}_{u}^{T}\\
				\mathcal{H}_{y}^{T}
			\end{bmatrix}\alpha =\chi_{L},\\
			&~~\quad~~\mathcal{H}_{u}^{k}\alpha \in \mathbb{U},~~\mathcal{H}_{y}^{k}\alpha \in \mathbb{Y},~~\forall k \in \mathcal{I}_{L},
		\end{align}
		where 
		\begin{align}
			& W_{d}=(\mathcal{H}_{u}^{F})'\mathcal{R}\mathcal{H}_{u}^{F}+(\mathcal{H}_{y}^{F})'\mathcal{Q}\mathcal{H}_{y}^{F},\label{eq:W}\\
			& c_{d}=-\left[(\mathcal{H}_{u}^{F})'\mathcal{R}u_{L}^{s}+(\mathcal{H}_{y}^{F})'\mathcal{Q}y_{L}^{s}\right],
		\end{align}
		with $\mathcal{Q}\!=\!\mathrm{diag}(Q,\ldots,Q) \succeq 0$, $\mathcal{R}\!=\!\mathrm{diag}(R,\ldots,R) \succ 0$ and $u_{L}^{s}$ and $y_{L}^{s}$ stacking $L$ copies of $u^{s}$ and $y^{s}$, respectively. 
	\end{subequations}
	However, the weighting matrix $W_{d}$ in \eqref{eq:W} can be shown to be positive \emph{semi-definite}, as illustrated in the proof of the following lemma.
	\begin{lemma}[Positive semi-definiteness of $W_{d}$]
		The matrix $W_{d}\!\in\!\mathbb{R}^{n_{\alpha} \times n_{\alpha}}$ in \eqref{eq:W} is positive semi-definite.
	\end{lemma}
	\begin{proof}
		Since $\mathcal{Q} \succeq 0$, then $(\mathcal{H}_{y}^{F})'\mathcal{Q}\mathcal{H}_{y}^{F}$ is positive semi-definite. As Theorem~\ref{thm:1} requires $N \geq (m+1)(L+n)-1$, then $n_{\alpha} \geq m(L+n)$. Since $\mathcal{H}_{u}^{F} \in \mathbb{R}^{mL \times n_{\alpha}}$ is full row rank by construction and $mL < n_{\alpha}$, then $(\mathcal{H}_{u}^{F})'\mathcal{R}\mathcal{H}_{u}^{F}$ is positive semi-definite, despite $\mathcal{R}\succ 0$. As such, $W_{d}$ is positive semi-definite, thus concluding the proof. 
	\end{proof}

	This structural property of $W_{d}$ makes the cost of the optimization problem in  \eqref{eq:alpha_MPC} convex, \emph{but not} strictly convex, ultimately hampering the possibility to retrieve a \emph{unique} explicit solution of the problem. To overcome this limitation, we \emph{augment} the cost with an $L_{2}$-regularization term, leading to the following regularized data-based predictive control problem:
	\begin{subequations}\label{eq:regularized_MPC}
		\begin{align}
			& \min_{\alpha}~~||\alpha||_{W_{d}}^{2}+2c_{d}'\alpha+\rho_{\alpha}\|\alpha\|^{2}\\ 
			&~~\mbox{s.t. } \begin{bmatrix}
				\mathcal{H}_{u}^{P}\\
				\mathcal{H}_{y}^{P}
			\end{bmatrix}\alpha=\chi_{0},\label{eq:init_constr1}\\
			&~~\quad~~\begin{bmatrix}
				\mathcal{H}_{u}^{T}\\
				\mathcal{H}_{y}^{T}
			\end{bmatrix}\alpha = \chi_{L},\label{eq:term_constr1}\\
			&~~\quad~~\mathcal{H}_{u}^{k}\alpha \in \mathbb{U},~~\mathcal{H}_{y}^{k}\alpha \in \mathbb{Y},~~\forall k \in \mathcal{I}_{L} \label{eq:value_constr1},
		\end{align} 
	\end{subequations}
	where $\rho_{\alpha}>0$ is an hyper-parameter to be tuned.
	\begin{remark}[Small $\rho_{\alpha}$]
		For sufficiently small values $\rho_{\alpha}$, the difference between problem \eqref{eq:nominalMPC} and \eqref{eq:regularized_MPC} can become negligible. In this case, R-EDDPC is likely to inherit the same property of the implicit solution (see \cite[Section III.B]{Berberich2021}).
	\end{remark}	 	
	\subsection{The role of regularization}
	By resulting in the addition of a positive constant to the diagonal of $W_{d}$, the regression term allows us to bypass the structural issue characterizing the cost in~\eqref{eq:cost_fcn}. Therefore, the larger $\rho_{\alpha}$, the more the regularized cost will differ from the one of problem~\eqref{eq:alpha_MPC}. One should thus pick a relatively small regularization parameter for the explicit solution of \eqref{eq:regularized_MPC} to be as close as possible to the one of the implicit MPC problem in \eqref{eq:alpha_MPC}.\\
	Even if the regularized problem~\eqref{eq:alpha_MPC} has been mainly introduced to allow for the computation of the explicit law, we stress that this alternative formulation of the implicit MPC problem goes in the direction of robustification, along the same line followed in \cite{Berberich2021,Coulson2019}. Indeed, $L_2$-regularization has a shrinkage effect on the (implicit) model of $\mathcal{P}$ embedded in $\alpha$. By penalizing the size of its components, the regularization terms steers the elements of $\alpha$ towards zero and, concurrently, towards each others. As such, its use $(i)$ prevents the model from becoming excessively complex, hence hindering overfitting, $(ii)$ it helps in handling noisy data, by implicitly reducing the influence of noise in the prediction accuracy, and $(iii)$ it alleviates problems that can be caused by highly correlated features. Therefore, in selecting $\rho_{\alpha}$ one has to further account for the shrinking effect of this additional penalty.
	
	\subsection{The derivation of R-EDDPC}\label{Sec:derivation}
	To derive the explicit solution of \eqref{eq:regularized_MPC}, let us firstly merge the constraints in \eqref{eq:init_constr1} and \eqref{eq:term_constr1} into a single equality $\mathcal{H}_{d}\alpha=\chi$. Thanks to the polytopic structure of the value constraints in \eqref{eq:value_constr1}, we can rewrite them as $\mathcal{G}_{d}\alpha \leq \beta$, so that the problem in \eqref{eq:regularized_MPC} can be equivalently recast as
	\begin{subequations}\label{eq:regMPC_compact}
		\begin{align}
			& \min_{\alpha}~~\frac{1}{2}||\alpha||_{W_{d}^{\rho}}^{2}+c_{d}'\alpha\\ 
			&~~\mbox{s.t. }~ \mathcal{H}_{d}\alpha =\chi,\\
			&~~\quad~~~\mathcal{G}_{d}\alpha\leq \beta \label{eq:ineq_constr2},
		\end{align}
		where the cost is scaled with respect to the one in \eqref{eq:regularized_MPC} and $W_{d}^{\rho}=W_{d}+\rho_{\alpha} I_{n_{\alpha}}$.
	\end{subequations}
	Moreover, let us introduce the following assumption on the constraints in \eqref{eq:ineq_constr2}.
	\begin{assumption}[Constraints]\label{assump:constraint_shape}
		Given $\mathcal{G}_{d}$ in \eqref{eq:ineq_constr2}, let its rows associated with active constraints be denoted by $\tilde{\mathcal{G}}_{d}$. The rows of 
		$\tilde{G}=\begin{bmatrix}\tilde{\mathcal{G}}_{d}' & \mathcal{H}_{d}'\end{bmatrix}'$ are linearly independent.
	\end{assumption}
	Under this assumption, the closed-form data-driven solution of the implicit problem in \eqref{eq:regMPC_compact} is given by the following theorem.
	\begin{theorem}[R-EDPPC]\label{thm:2}
		Let Assumptions~\ref{assump:1}-\ref{assump:constraint_shape} hold, with the latter satisfied for all possible combinations of active constraints $M \in \mathbb{N}$. Then, the data-driven explicit law coupled with \eqref{eq:regMPC_compact} is unique and is given by
		\begin{equation}\label{eq:dd_law}
			u(\chi)=\begin{cases}
				F_{d,1}\chi+\mathrm{f}_{d,1}~~\mbox{ if }~~E_{d,1}\chi \leq K_{d,1},\\
				\vdots\\
				F_{d,M}\chi+\mathrm{f}_{d,M}~~\mbox{ if }~~E_{d,M}\chi \leq K_{d,M}.			 
			\end{cases}
		\end{equation}
	\end{theorem}
	\begin{proof}
		Since problem~\eqref{eq:regMPC_compact} is strictly convex by construction, it has a unique solution $\alpha$, whose closed-form can be found by applying the Karush-Kuhn-Tucker (KKT) conditions. Therefore, the following holds:
		\begin{subequations}
			\begin{align}
				& W_{d}^{\rho}\alpha+\mathcal{G}_{d}'\lambda+\mathcal{H}_{d}'\mu+c_{d}=0, \label{eq:KKT1}\\
				& \lambda'(\mathcal{G}_{d}\alpha-\beta)=0, \label{eq:KKT2}\\
				& \mathcal{H}_{d}\alpha-\chi=0, \label{eq:KKT3}\\
				& \mathcal{G}_{d}\alpha-\beta \leq 0, \label{eq:KKT4}\\
				& \lambda \geq 0, \label{eq:KKT5} 
			\end{align} 
		\end{subequations}
		where $\lambda \in \mathbb{R}^{(m+p)L}$ and $\mu \in \mathbb{R}^{2n(m+p)}$ are the Lagrange multipliers associated with the inequality and equality constraints in \eqref{eq:regMPC_compact}, respectively.\\
		Consider a generic set of active constraints in \eqref{eq:ineq_constr2}, and let us distinguish the Lagrange multipliers $\tilde{\lambda} \in \mathbb{R}^{n_{\lambda}}$ associated with the latter and the ones coupled with inactive constraints, here indicated as $\hat{\lambda}$. Accordingly, we denote by $\tilde{\mathcal{G}}_{d}$ and $\tilde{\beta}$ the rows of $\mathcal{G}_{d}$ and $\beta$ associated with active constraints, while we indicate with $\hat{\mathcal{G}}_{d}$, $\hat{\beta}$ the remaining rows.\\
		Since both \eqref{eq:KKT2} and the dual feasibility condition in \eqref{eq:KKT5} have to hold for inactive constraints, $\hat{\lambda}$ is a vector of zeros and, thus, the KKT conditions in \eqref{eq:KKT2} and \eqref{eq:KKT4} can be equivalently restated as:
		\begin{subequations}
			\begin{align}
				&\tilde{\mathcal{G}}_{d}\alpha-\tilde{\beta}=0, \label{eq:KKT2_2}\\
				&\hat{\mathcal{G}}_{d}\alpha - \hat{\beta} < 0, \label{eq:KKT4_2}	
			\end{align}
			where the product with $\tilde{\lambda}'$ in \eqref{eq:KKT2_2} is neglected, since this vector is non-zero by definition.
		\end{subequations}
		Thanks to this reformulation, we can then merge \eqref{eq:KKT3} and \eqref{eq:KKT2_2}, so as to obtain a single equality condition, here defined as
		\begin{equation}\label{eq:merged_equality}
			\tilde{G}\alpha-b-\tilde{S}\chi=0,
		\end{equation}  
		where
		\begin{equation*}
			\tilde{G}= \begin{bmatrix}
				\tilde{\mathcal{G}}_{d}\\
				\mathcal{H}_{d}
			\end{bmatrix},~~b=\begin{bmatrix}
				\tilde{\beta} \\
				0
			\end{bmatrix},~~\tilde{S}=\begin{bmatrix}
				0\\
				I_{n(m+p)}
			\end{bmatrix},
		\end{equation*}
		and the dimensions of $\tilde{S}$ depend on the number of active constraints. We can then exploit the stationary condition in \eqref{eq:KKT1} to express $\alpha$ as a function of the non-zero Lagrange multipliers, \emph{i.e.,}
		\begin{equation}\label{eq:alpha_lag}
			\alpha=-(W_{d}^{\rho})^{-1}\left(\tilde{G}'\delta+c_{d}\right),
		\end{equation}
		where $\delta$ is given by $\delta=\smallmat{
			\tilde{\lambda}' &
			\mu'	
		}'$. By merging \eqref{eq:alpha_lag} and \eqref{eq:merged_equality}, we then obtain the equality
		\begin{equation}\label{eq:merged2}
			-\tilde{G}_W\delta-\tilde{b}-\tilde{S}\chi=0
		\end{equation} 
		where $\tilde{G}_{W}=\tilde{G}(W_{d}^{\rho})^{-1}\tilde{G}'$ and $\tilde{b}=\tilde{G}(W_{d}^{\rho})^{-1}c_{d}+b$. We can thus retrieve the closed-form expression  for $\delta$, which is given by
		\begin{equation}\label{eq:delta}
			\delta = -\tilde{G}_{W}^{-1}(\tilde{b}+\tilde{S}\chi).
		\end{equation}
		Since $\tilde{G}_{W}$ can always be inverted according to our assumptions, we can thus substitute \eqref{eq:delta} into \eqref{eq:alpha_lag}, ultimately obtaining the following data-driven expression for $\alpha$:
		\begin{equation}\label{eq:alpha_closedform}
			\alpha=(W_{d}^{\rho})^{-1}\left[\tilde{G}'\tilde{G}_{W}^{-1}(\tilde{b}+\tilde{S}\chi)-c_{d}\right].
			\end{equation}
		Accordingly, we can retrieve the associated data-driven expression for the predicted input sequence for a given set of active constraints as
		\begin{equation}\label{eq:datadriven_input}
			\tilde{u}_{[0,L-1]}(\chi)\!=\!\mathcal{H}_{u}^{F}(W_{d}^{\rho})^{-1}\!\left[\tilde{G}'\tilde{G}_{W}^{-1}(\tilde{b}\!+\!\tilde{S}\chi)\!-\!c_{d}\right]\!\!.
		\end{equation}
		By combining the primal and dual feasibility conditions in \eqref{eq:KKT4} and \eqref{eq:KKT5}, we can further characterize the polyhedral region where \eqref{eq:datadriven_input} holds, which is shaped by the following combination of inequalities:
		\begin{subequations}\label{eq:polyhedral_regions}
			\begin{align}
				& \mathcal{G}_{d}(W_{d}^{\rho})^{-1}\!\!\left[\tilde{G}'\tilde{G}_{W}^{-1}(\tilde{b}\!+\!\tilde{S}\chi)\!-\!c_{d}\right]\!\!-\!\beta \!\leq 0,\\
				& [\tilde{G}_{W}^{-1}(\tilde{b}+S\chi)]_{1:n_{\lambda}}\leq 0,
			\end{align}
		\end{subequations}
		where $n_{\lambda}$ denotes the number of active constraints.\\
		By considering all $M$ possible combinations of active constraints, straightforward manipulations result in an explicit predictive control sequence defined as
		\begin{subequations}\label{eq:predicted_sequence}
			\begin{equation}
				\bar{u}_{[0,L-1]}\!=\!\begin{cases}
					\mathcal{F}_{d,1}\chi+f_{d,1}\mbox{ if }~\mathcal{E}_{d,1}\chi \leq \mathcal{K}_{d,1},\\
					\vdots\\
					\mathcal{F}_{d,M}\chi+f_{d,M}\mbox{ if }~\mathcal{E}_{d,M}\chi \leq \mathcal{K}_{d,M},				 
				\end{cases}
			\end{equation}	
			where
			\begin{align}
				& \mathcal{F}_{d,i} = \mathcal{H}_{u}^{F}(W_{d}^{\rho})^{-1}\tilde{G}_{i}'\tilde{G}_{W,i}^{-1}\tilde{S}_{i},\\
				& f_{d,i} = \mathcal{H}_{u}^{F}(W_{d}^{\rho})^{-1}\left(\tilde{G}_{i}'\tilde{G}_{W,i}^{-1}\tilde{b}_{i}-c_{d,i}\right),\\
				& \mathcal{E}_{d,i} = \begin{bmatrix}
					\mathcal{G}_{d}(W_{d}^{\rho})^{-1}\tilde{G}_{i}'\tilde{G}_{W,i}^{-1}\tilde{S}_{i},\\
					\tilde{G}_{W,i}^{-1}\tilde{S}_{i}
				\end{bmatrix},\\
				& \mathcal{K}_{d,i} = \begin{bmatrix}
					\beta_{i}+\mathcal{G}_{d}(W_{d}^{\rho})^{-1}\left(c_{d,i}-\tilde{G}_{i}'\tilde{G}_{W,i}^{-1}\tilde{b}_{i}\right)\\
					-\tilde{G}_{W,i}^{-1}\tilde{b}_{i}
				\end{bmatrix}, 
			\end{align}
			for all $i=1,\ldots,M$, where $\tilde{G}_{i}$, $\tilde{b}_{i}$, $c_{d,i}$ and $\tilde{S}_{i}$ indicate the rows of $\tilde{G}$, $\tilde{b}$, $c_{d}$ and $\tilde{S}$ associated with the $i$-th set of active constraint, and  $\tilde{G}_{W,i}$ is constructed accordingly.
		\end{subequations}
		Lastly, by selecting the first element of this sequence, it can easily be shown that the control action to be applied is indeed piecewise affine and that it has the structure in \eqref{eq:dd_law}, thus concluding the proof.
	\end{proof}

	Note that, the explicit control law in \eqref{eq:dd_law} has a structure similar to the one obtained in the standard model-based case \cite{Alessio2009,Bemporad2002b}, but the dependence on the matrices of the state-space model of $\mathcal{P}$ have now been replaced with a set of data matrices. Moreover, differently from the explicit predictive law introduced in \cite{Sassella2021}, the obtained piecewise affine law is inherently output-feedback, thus not requiring a direct measurement of the state.
\begin{remark}[Relaxation] Assumption~\ref{assump:constraint_shape} can be relaxed in practice, since redundant constraints can be removed with degeneracy handling strategies, like the ones proposed in \cite{Bemporad2001}.   
		\end{remark}
	\begin{remark}[Coping with general constraints]
	The data-driven law in \eqref{eq:dd_law} can be readily adapted to handle more general polytopic constraints of the form
		\begin{equation}
			\mathcal{G}_{d}\alpha \leq \beta+\varphi \chi.
		\end{equation}
		In this case, the definition of $\tilde{S}$ in \eqref{eq:merged_equality} changes as follows:
		\begin{equation*}
			\tilde{S}=\begin{bmatrix}
				\tilde{\varphi}\\
				I_{n(m+p)}
			\end{bmatrix},
		\end{equation*} 
		with $\tilde{\varphi}$ indicating the rows of $\varphi$ associated to the considered set of active constraints. Nonetheless, the derivations in the proof of Theorem~\ref{thm:2} remain unchanged, and so does the form of the explicit control law.
	\end{remark}	
	\section{A comparison with E-DDPC}\label{Sec:comparison}
	The explicit law derived in Section~\ref{Sec:derivation} is not the first of its kind. Indeed, a data-driven explicit law has already been proposed in \cite{Sassella2021}. Our aim is thus to compare the two DD-PC problems that R-EDDPC and E-DDPC solve.\\
	To this end, let us recall the implicit problem considered in \cite{Sassella2021} by focusing on the case in which the prediction, control and constraint horizons are all equal to $L$, \emph{i.e.,}
	\begin{subequations}\label{eq:EDDPC_prob}
		\begin{align}
			& \min_{\bar{u}_{[0,L-1]}} \sum_{k=0}^{L-1} \left[\|\bar{x}_{k}\|_{\tilde{Q}}^{2}+\|\bar{u}_{k} \|_{\tilde{R}}^{2}\right]+\|\bar{x}_{L}\|_{\tilde{P}}^{2} \label{eq:EDDPC_cost}\\
			&~~\mbox{s.t. }~\bar{x}_{k+1}\!=\!X_{1,N}\!\begin{bmatrix}
				U_{0,1,N}\\
				\hline
				X_{0,N}
			\end{bmatrix}^{\!\dagger} \!\begin{bmatrix}
				\bar{u}_{k}\\
				\bar{x}_{k}
			\end{bmatrix}\!, ~k\!=\!0,\ldots,L\!-\!1, \label{eq:EDDPC_model}\\
			& \quad \quad ~~\bar{x}_0 = x,  \label{eq:EDDPC_init}\\
			& \quad \quad ~~\bar{u}_{k} \in \mathbb{U}, ~~ \bar{x}_{k} \in \mathbb{X},~k=0,\ldots,L-1,  \label{eq:EDDPC_value}
		\end{align}
	\end{subequations}
	where $x \in \mathbb{R}^{n}$ denotes the initial state for the prediction at time $t$ and
	\begin{align}
		&U_{0,1,N}=\begin{bmatrix}
			u_{n}^{d} & u_{n+1}^{d} & \cdots & u_{N+n-1}^{d}
		\end{bmatrix},\\
		&X_{0,N}= \begin{bmatrix}
			x_{n}^{d} & x_{n+1}^{d} & \cdots & x_{N+n-1}^{d}
		\end{bmatrix},\label{eq:X0}\\
		& X_{1,N}=\begin{bmatrix}
			x_{n+1}^{d} & x_{n+2}^{d} & \cdots & x_{N+n}^{d}
		\end{bmatrix},\label{eq:X1}
	\end{align}
	with $x_{k}^{d}$ denoting the state measured or reconstructed from data at instant $k$. Note that, since the input is persistently exciting of order $L+2n$ according to Assumption~\ref{assump:1}, the condition required for the model in \eqref{eq:EDDPC_model} to represent the behavior of the unknown system $\mathcal{P}$ holds when the state is fully measured (see \cite[Theorem 1]{Sassella2021}). \\
	From \eqref{eq:nominalMPC} and \eqref{eq:EDDPC_prob} to be comparable, when the state is not fully measured, we solve problem~\eqref{eq:EDDPC_prob} by considering the non-minimal state realization
		\begin{equation}\label{eq:nonminimal}
			\bar{z}_{k}=\begin{bmatrix}
				u_{k-n}'  & \cdots & u_{k-1}' & y_{k-n}' & \cdots & y_{k-1}
			\end{bmatrix}'.
		\end{equation}
	Note that, for the predictive model in \eqref{eq:EDDPC_model} to be well defined in this scenario, the input has to be persistently exciting of order $2n+1$ \cite[Theorem 7]{DePersis2019}. In our setting, this condition still holds thanks to Assumption~\ref{assump:1}. By relying on the above problem setting, we can prove the existence of the following equivalence relations.
	
	\begin{lemma}[Model equivalence]\label{lemma:1}
		The predictive models in \eqref{eq:model} and \eqref{eq:EDDPC_model} are equivalent.
	\end{lemma}
	
	\begin{proof}
		See Appendix~\ref{appendix:A}
	\end{proof}	
	
%	As a direct consequence of Lemma~\ref{lemma:1}, we can further show that the constraints in \eqref{eq:EDDPC_init}-\eqref{eq:EDDPC_value} are equivalent to the ones in \eqref{eq:initial_constraint} and in \eqref{eq:feasibility_condition}, as outlined in the following lemma. 
	
	\begin{lemma}[Constraints equivalence]\label{lemma:2}
		The constraints in \eqref{eq:initial_constraint} and \eqref{eq:EDDPC_init} are always equivalent. Instead, the ones in \eqref{eq:feasibility_condition} and \eqref{eq:EDDPC_value} are equivalent if: $(i)$ $\mathbb{X} \equiv \mathbb{Y}$ when the state is fully measured; $(ii)$ $\mathbb{X} \equiv \mathbb{U}^{n} \times \mathbb{Y}^{n}$ otherwise.
	\end{lemma}
	
	\begin{proof}
		See Appendix~\ref{appendix:B}.
	\end{proof}
	
	Since \eqref{eq:EDDPC_prob} aims at steering both states and inputs to zero, while \eqref{eq:nominalMPC} depends on the equilibrium $(u^{s},y^{s})$, we make the comparison for  $(u^{s},y^{s})=(0,0)$.	We can now derive the following result on the relationship between \eqref{eq:nominalMPC} and \eqref{eq:EDDPC_prob}.	
	\begin{theorem}[Problem equivalence]\label{thm:4}
		Consider the relaxed data-driven problem
		\begin{subequations}\label{eq:relaxed_MPC}
			\begin{align}
				&\min_{\alpha,\bar{u},\bar{y}}~ \sum_{k=0}^{L-1} \left[\|\bar{y}_{k}\|_{Q}^{2}\!+\!\| \bar{u}_{k}\|_{R}^{2}\right]\!+\!\bigg\| \begin{bmatrix}
					\bar{u}_{[L-n,L-1]}\\
					\bar{y}_{[L-n,L-1]}	
				\end{bmatrix} \bigg\|_{P}^{2}\!\!\!\!\!+\!\rho_{\alpha}\|\alpha\|^{2}\!\!\!\!\!\\
				&~~~\mbox{s.t. } \begin{bmatrix}
					\bar{u}_{[-n,L-1]}\\
					\bar{y}_{[-n,L-1]}
				\end{bmatrix}=\begin{bmatrix}
					H_{L+n}(u^{d})\\
					H_{L+n}(y^{d})
				\end{bmatrix}\alpha,\\
				&~~~\quad~~ \begin{bmatrix}
					\bar{u}_{[-n,-1]}\\
					\bar{y}_{[-n,-1]}	
				\end{bmatrix}=\chi_{0},\\
				&~~~\quad~~~ \bar{u}_{k} \in \mathbb{U}, ~~ \bar{y}_{k} \in \mathbb{Y}~~\forall k \in \mathcal{I}_{L},
			\end{align}
		\end{subequations}
	where the terminal constraint in \eqref{eq:terminal_constraint} has been softened and added to the cost. Then, \eqref{eq:EDDPC_prob} and \eqref{eq:relaxed_MPC} are equivalent in the following cases:
		\begin{itemize}
			\item[$(i)$] the state is fully measurable, $\rho_{\alpha} \to 0$ and the weights in \eqref{eq:EDDPC_cost} are $Q=\tilde{Q}$, $R=\tilde{R}$, $P=T'\tilde{P}T$, where $T \in \mathbb{R}^{n \times (m+p)n}$ satisfies
			\begin{equation}\label{eq:aid_matrix}
				\bar{x}_{L}=T\begin{bmatrix}
					\bar{u}_{[L-n,L-1]}\\
					\bar{x}_{[L-n,L-1]}
				\end{bmatrix},
			\end{equation}
			\item[$(ii)$] the state is not fully measured, $\rho_{\alpha} \to 0$ and the weights in \eqref{eq:EDDPC_cost} are
			\begin{equation}\label{eq:weights}
				\tilde{Q}=V'\begin{bmatrix}
					R & 0\\
					0 &Q
				\end{bmatrix}V,~~\tilde{R}=0,~~\tilde{P}=P+\tilde{Q},
			\end{equation}
			with $V \in \mathbb{R}^{(m+p)\times n(m+p)}$ verifying 
			\begin{equation}
				\begin{bmatrix}
					\bar{u}_{k-1}\\
					\bar{y}_{k-1}
				\end{bmatrix}=V\bar{z}_{k}.
			\end{equation}
		\end{itemize}
	\end{theorem}
	
	\begin{proof}
		See Appendix~\ref{appendix:C}.
	\end{proof}
	
	Since the steps performed to compute R-EDDPC and E-DDPC are fundamentally the same, the equivalence between \eqref{eq:nominalMPC} and \eqref{eq:relaxed_MPC} shows that the two explicit law are likely to match when $\rho_{\alpha}$ vanishes and the terminal constraint is lifted to the cost in \eqref{eq:regularized_MPC} or, alternatively, the terminal cost is replaced with a hard constraint in \eqref{eq:EDDPC_prob}.
	
	\section{Extensions}\label{Sec:extensions}
	We now highlight how the explicit solution derived in Section \ref{Sec:derivation} can be extended to two different scenarios. Firstly, we show how the presented derivation can be adapted to handle tracking tasks, according to the scheme proposed in \cite{Berberich2020b}. Then, along the line followed in \cite{Berberich2021}, we introduce a slack variable to further robustify the approach with respect to noisy data and discuss how this modifies the obtained explicit solution.
	\subsection{Handling changing set points in R-EDDPC}
	When the control objectives shifts from reaching a given equilibrium point $(u^{s},y^{s})$ to tracking an input/output reference behavior ($u^{r},y^{r}$), the implicit MPC problem to be solved at each time step $t \in \mathbb{N}$ can be modified as follows \cite{Berberich2020b}:
	\begin{subequations}\label{eq:tracking_MPC}
		\begin{align}
			&\min_{\substack{\alpha,\bar{u},\bar{y}\\u^{s},y^{s}}}~~\tilde{\ell}(\bar{u}_{k},\bar{y}_{k},u^{s},y^{s})\\
			&~~~\mbox{s.t. } \begin{bmatrix}
				\bar{u}_{[-n,L-1]}\\
				\bar{y}_{[-n,L-1]}
			\end{bmatrix}=\begin{bmatrix}
				H_{L+n}(u^{d})\\
				H_{L+n}(y^{d})
			\end{bmatrix}\alpha,\label{eq:model_tracking}\\
			&~~~\quad~~ \begin{bmatrix}
				\bar{u}_{[-n,-1]}\\
				\bar{y}_{[-n,-1]}	
			\end{bmatrix}=\chi_{0}, \label{eq:initial tracking} \\
			&~~~\quad ~~ \begin{bmatrix}
				\bar{u}_{[L-n,L-1]}\\
				\bar{y}_{[L-n,L-1]}	
			\end{bmatrix}=\Omega\begin{bmatrix}u^{s}\\
				y^{s}\end{bmatrix},		 \label{eq:terminal tracking}\\
			&~~~\quad~~~ \bar{u}_{k} \in \mathbb{U}, ~~ \bar{y}_{k} \in \mathbb{Y}~~\forall k \in \mathcal{I}_{L}, \label{eq:tracking_feasibility}\\
			&~~~\quad ~~ (u^{s},y^{s}) \in \mathbb{U}^{s} \times \mathbb{Y}^{s}, \label{eq:terminal_feasibility}
		\end{align}
	\end{subequations} 
	where the cost is now given by
	\begin{align*}
		\tilde{\ell}(\bar{u}_{k},\bar{y}_{k},u^{s},y^{s})=&\sum_{k=0}^{L-1}\|\bar{u}_{k}-u^{s}\|_{R}^{2}+\|\bar{y}_{k}-y^{s}\|_{Q}^{2}\\
		&~+\|u^{s}\!-{u}^{r}\|_{\Psi}^{2}\!+\!\|y^{s}\!-y^{r}\|_{\Phi}^{s}+\rho_{\alpha}\|\alpha\|^{2}\!,
	\end{align*}
	so as to penalize the deviation of the set point $(u^{s},y^{s})$ from the desired target $(u^{r},y^{r})$, with $\Psi,\Phi \succ 0$. Note that, $\Omega \in \{0,1\}^{n(m+p) \times m+p}$ in \eqref{eq:terminal tracking} is such that:  
	\begin{equation*}
		\Omega \begin{bmatrix}
			u^{s}\\
			y^{s}
		\end{bmatrix} =\chi^{T}.
	\end{equation*}
	Moreover, the constraint in \eqref{eq:terminal_feasibility} entails some value conditions on the equilibrium point, here assumed to be still characterized via a set of polytopic constraints.
	
	Within this scenario, the resulting explicit predictive law can be computed by relying on the following lemma.
	
	\begin{lemma}[Piecewise affine solution]\label{thm:3}
		Let $\bar{\alpha}$ collect the optimization variables of problem~\eqref{eq:tracking_MPC} and $\bar{\chi}$ stack the initial conditions $\chi_{0}$ and the reference to be tracked, \emph{i.e.,}
		\begin{equation}\label{eq:extended_alpha}
			\bar{\alpha}=\begin{bmatrix}
				\alpha\\
				u^{s}\\
				y^{s}
			\end{bmatrix},~~~\bar{\chi}=\begin{bmatrix}
				\chi_{0}\\
				u^{r}\\
				y^{r}
			\end{bmatrix}. 
		\end{equation}
		Then, problem~\eqref{eq:tracking_MPC} is explicitly solved by a piecewise affine law defined as in \eqref{eq:dd_law}, with $\bar{\chi}$ substituting $\chi$. 
	\end{lemma}
	
	\begin{proof}
	Based on the definition of $\bar{\alpha}$ in \eqref{eq:extended_alpha} we can rewrite the constraint in \eqref{eq:model_tracking} as
		\begin{equation}
			\begin{bmatrix}
				\bar{u}_{[-n,L-1]}\\
				\bar{y}_{[-n,L-1]}
			\end{bmatrix}\!=\!\begin{bmatrix}
				H_{L+n}(u^{d}) & 0\\
				H_{L+n}(y^{d}) & 0
			\end{bmatrix}\bar{\alpha}.
		\end{equation}
		This equivalent representation allows us to recast the tracking MPC problem \eqref{eq:regMPC_compact} as
		\begin{subequations}
			\begin{align}
				& \min_{\bar{\alpha}}~~\frac{1}{2}||\bar{\alpha}||_{\bar{W}_{d}^{\rho}}^{2}+\bar{\chi}'\bar{c}_{d}\bar{\alpha}\\ 
				&~~\mbox{s.t. }~ \bar{\mathcal{H}}_{d}\bar{\alpha} =\bar{\mathcal{S}}\bar{\chi},\\
				&~~\quad~~~\bar{\mathcal{G}}_{d}\bar{\alpha}\leq \bar{\beta} \label{eq:ineq_constr3},
			\end{align}
		\end{subequations}
		where \eqref{eq:ineq_constr3} is obtained by merging \eqref{eq:tracking_feasibility}-\eqref{eq:terminal_feasibility}. Note that the matrices characterizing this equivalent problem are defined as follows:
		\begin{align*}
			&\bar{W}_{d}^{\rho}=(\bar{\mathcal{H}}_{u}^{F})'\mathcal{R}\bar{\mathcal{H}}_{u}^{F}+(\bar{\mathcal{H}}_{y}^{F})'\mathcal{Q}\bar{\mathcal{H}}_{y}^{F}\\
			&\qquad\quad+(S_{u^{s}})'\Psi S_{u^{s}}+(S_{y^{s}})'\Phi S_{y^{s}}+\rho_{\alpha}I_{},\\
			&\bar{c}_{d}=-\begin{bmatrix}
				0 & 0 & 0\\
				0 & \Psi  & 0\\
				0 & 0 & \Phi 
			\end{bmatrix},\\
			&\bar{\mathcal{H}}_{d}=\begin{bmatrix}
				\mathcal{H}^{P} & ~~0\\
				\mathcal{H}^{T} & -\Omega
			\end{bmatrix},~~\bar{\mathcal{S}}=\begin{bmatrix}
				I & 0 & 0\\
				0 & 0 & 0
			\end{bmatrix},
		\end{align*}
		with \begin{align*}
		&\bar{\mathcal{H}}_{u}^{F}=\begin{bmatrix}
				\mathcal{H}_{u}^{F} & -\mathcal{C}_{L}^{u^s} & 0
			\end{bmatrix}, ~~ \bar{\mathcal{H}}_{y}^{F}=\begin{bmatrix}
				\mathcal{H}_{y}^{F} & 0 & -\mathcal{C}_{L}^{y^s}
			\end{bmatrix},\\
			& \mathcal{H}^{P}=\begin{bmatrix}
				\mathcal{H}_{u}^{P},\\
				\mathcal{H}_{y}^{P}
			\end{bmatrix},~~\mathcal{H}^{T}=\begin{bmatrix}
				\mathcal{H}_{u}^{T}\\
				\mathcal{H}_{y}^{T}
			\end{bmatrix},
		\end{align*}
		and where $S_{\gamma}$ selects any variable (denoted generically with the placeholder $\gamma$) within a given vector, while $\mathcal{C}_{L}^{\gamma}$ allows one to construct $L$ copies of it. This reformulation allows us to follow the same step presented in Section~\ref{Sec:derivation} to derive the explicit law, by replacing $\tilde{S}$ in \eqref{eq:merged_equality} with
		\begin{equation*}
			\tilde{S}=\begin{bmatrix}
				0\\
				\bar{\mathcal{S}}
			\end{bmatrix},
		\end{equation*}
		and $c_d$ in \eqref{eq:alpha_lag}-\eqref{eq:polyhedral_regions} with $\bar{c}_{d}\bar{\chi}$. As a consequence, the predicted input sequence has the same closed-form in \eqref{eq:predicted_sequence}, with
		\begin{align*}
			& \mathcal{F}_{d,i} = \bar{\mathcal{H}}_{u}^{F}(\bar{W}_{d}^{\rho})^{-1}\left(\tilde{G}_{i}'\tilde{G}_{W,i}^{-1}\tilde{S}_{i}-\bar{c}_{d,i}\right),
			\\
			& f_{d,i} = \bar{\mathcal{H}}_{u}^{F}(\bar{W}_{d}^{\rho})^{-1}\tilde{G}_{i}'\tilde{G}_{W,i}^{-1}\tilde{b}_{i},
			\\
			& \mathcal{E}_{d,i} = \begin{bmatrix}
				\bar{\mathcal{G}}_{d}(\bar{W}_{d}^{\rho})^{-1}\left(\tilde{G}_{i}'\tilde{G}_{W,i}^{-1}\tilde{S}_{i}+\bar{c}_{d,i}\right)\\
				\tilde{G}_{W,i}^{-1}\tilde{S}_{i}
			\end{bmatrix},
			\\
			& \mathcal{K}_{d,i} = \begin{bmatrix}
				\bar{\beta}_{i}-\bar{\mathcal{G}}_{d}(\bar{W}_{d}^{\rho})^{-1}\tilde{G}_{i}'\tilde{G}_{W,i}^{-1}\tilde{b}_{i}\\
				-\tilde{G}_{W,i}^{-1}\tilde{b}_{i}
			\end{bmatrix},
		\end{align*}
		where $\tilde{G}$, $\tilde{G}_{W}$ and $\tilde{b}$ can be readily customized to the current problem from the matrices introduced in Section~\ref{Sec:derivation}, while the subscript indicates the $i$-th set of active constraints. We can then retrieve the input by extracting the first component of this sequence only. This result leads to a law that takes the same piecewise affine form as \eqref{eq:dd_law}, thus concluding the proof. 
	\end{proof}
	
	\subsection{Robustification with slack variables}
	Assume that the measured outputs $y^{d}$ used to construct the Hankel matrices in \eqref{eq:model} are corrupted by noise. In this case, a robust DD-PC formulation similar to the one proposed in \cite{Berberich2021} can be recovered by introducing an additive slack variable $\sigma \in \mathbb{R}^{p(L+n)}$ on the predicted output and the corresponding regularization term as follows:  
	\begin{subequations}\label{eq:robustMPC}
		\begin{align}
			&\min_{\alpha,\bar{u},\bar{y}}~~ \sum_{k=0}^{L-1} \ell(\bar{u}_{k},\bar{y}_{k})+\rho_{\alpha}\|\alpha\|^{2}+\rho_{\sigma}\|\sigma\|^{2}\label{eq:cost_fcnR}\\
			&~~~\mbox{s.t. } \begin{bmatrix}
				\bar{u}_{[-n,L-1]}\\
				\bar{y}_{[-n,L-1]}+\sigma
			\end{bmatrix}=\begin{bmatrix}
				H_{L+n}(u^{d})\\
				H_{L+n}(y^{d})
			\end{bmatrix}\alpha,\label{eq:modelR}\\
			&~~~\quad~~ \begin{bmatrix}
				\bar{u}_{[-n,-1]}\\
				\bar{y}_{[-n,-1]}
			\end{bmatrix}=\chi_{0},\label{eq:robust_init}\\
			&~~~\quad ~~ \begin{bmatrix}
				\bar{u}_{[L-n,L-1]}\\
				\bar{y}_{[L-n,L-1]}	
			\end{bmatrix}=\chi_{L}, \\
			&~~~\quad~~~ \bar{u}_{k} \in \mathbb{U}, ~~ \bar{y}_{k} \in \mathbb{Y}~~\forall k \in \mathcal{I}_{L}.
		\end{align}
	\end{subequations}
	Along the line of \cite[Remark 3]{Berberich2021}, we do not include any constraint on the slack variable, by leveraging on the fact that its values can be practically contained by selecting $\rho_{\sigma}>0$ sufficiently large. In turn, this entails that \eqref{eq:robustMPC} can be formulated without any prior information on the measurement noise features, while accounting for the possible mismatch between the outputs predicted from noisy data and the true one.\\  
	To derive the explicit predictive law for this robust formulation, let us introduce the extended optimization variable
	\begin{equation}\label{eq:extended_variable}
		\tilde{\alpha}=\begin{bmatrix}
			\alpha\\
			\sigma
		\end{bmatrix},
	\end{equation}
	and modify the constraint in \eqref{eq:modelR} as
	\begin{equation}
		\begin{bmatrix}
			\bar{u}_{[-n,L-1]}\\
			\bar{y}_{[-n,L-1]}	
		\end{bmatrix}\!\!=\!\!\begin{bmatrix}
			H_{L+n}(u^{d}) & 0\\
			H_{L+n}(y^{d}) & -I_{L+n}
		\end{bmatrix}\!\tilde{\alpha}\!=\!\!\begin{bmatrix}
			\tilde{H}_{L+n}(u^{d})\\
			\tilde{H}_{L+n}(y^{d})
		\end{bmatrix}\tilde{\alpha},
	\end{equation} 
	accordingly. By redefining the matrices in \eqref{eq:aid_matrices} as
	\begin{subequations}\label{eq:aid_matricesR}
		\begin{align}
			& \tilde{\mathcal{H}}_{\gamma}^{P}\!=\![\tilde{H}_{L+n}(\gamma^{d})]_{1:n}, ~~ \tilde{\mathcal{H}}_{\gamma}^{F}\!=\![\tilde{H}_{L+n}(\gamma^{d})]_{n+1:L+n}, \\
			& \tilde{\mathcal{H}}_{\gamma}^{T}\!=\![\tilde{H}_{L+n}(\gamma^{d})]_{L+1:L+n},~ \tilde{\mathcal{H}}_{\gamma}^{k}\!=\![\tilde{H}_{L+n}(\gamma^{d})]_{k+1}, 
		\end{align}
		where $\gamma$ is still a placeholder, we can retrieve the explicit predictive law by following the same steps described in Sections~\ref{Sec:fromIMPtoEXP}-\ref{Sec:derivation}.  
	\end{subequations}
	
%	\begin{remark}[Use bounds to select penalties]
%		If a bound on the measurement noise $\bar{\varepsilon}$ is known, it can be exploited to scale the weight $\rho_{\alpha}$, as in \cite{Berberich2021}. In this case, we can redefine the latter as $\rho_{\alpha}=\varepsilon\tilde{\rho}_{\alpha}$ and tune $\tilde{\rho}_{\alpha}$ instead.	
%	\end{remark}

	\section{Benchmark numerical examples}\label{Sec:examples}
	In this section, we analyze the performance of R-EDDPC on two benchmark examples. Initially, we consider the problem introduced in \cite[Section 7.1]{Bemporad2002b}, with the plant $\mathcal{P}$ modified so as to force the whole state to be measurable. This choice allows us to compare the performance of R-EDDPC with the one of E-DDPC. We then consider the example introduced in \cite{Berberich2021}. In this case, we juxtapose the results attained by R-EDDPC with the ones achieved designing the explicit MPC law using the standard two-stage procedure, \emph{i.e.,} by identifying a state-space model of the system first. We stress that, independently of the considered example, R-EDDPC is always designed by relying on data only. All computations have been carried out on an Intel Core i7-7700HQ processor, running MATLAB 2019b.
	\subsection{SISO system with fully measurable state}
	\begin{figure}[!tb]
		\begin{center}
			\begin{tabular}{c}
				\subfigure[First state component]{\includegraphics[scale=.65,trim=1.5cm 0.5cm 9.5cm 22.5cm,clip]{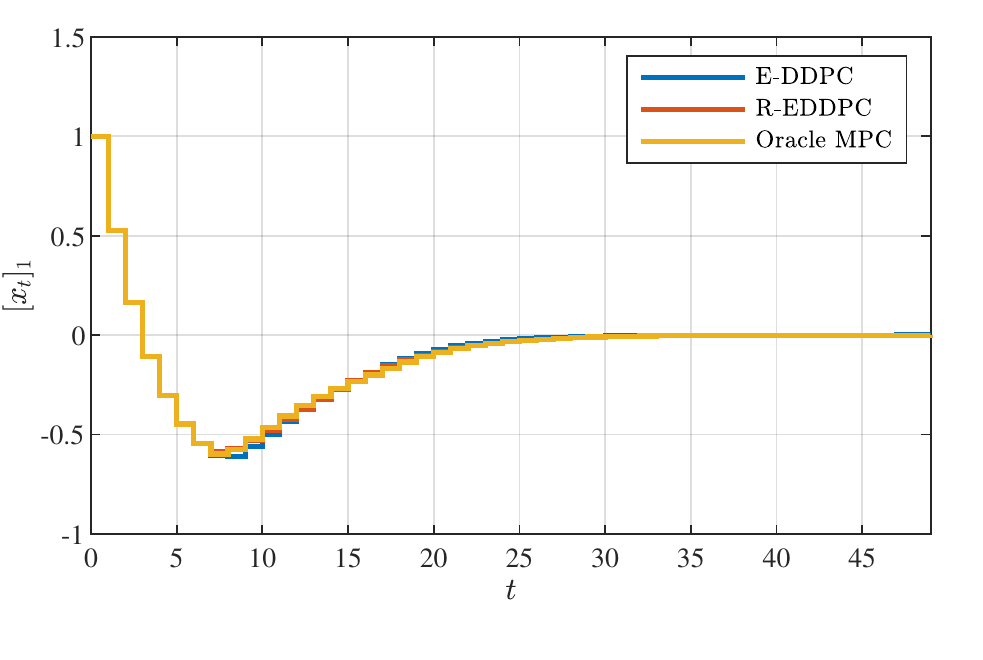}}\\
				\subfigure[Second state component]{\includegraphics[scale=.65,trim=1.5cm .5cm 9.5cm 22.5cm,clip]{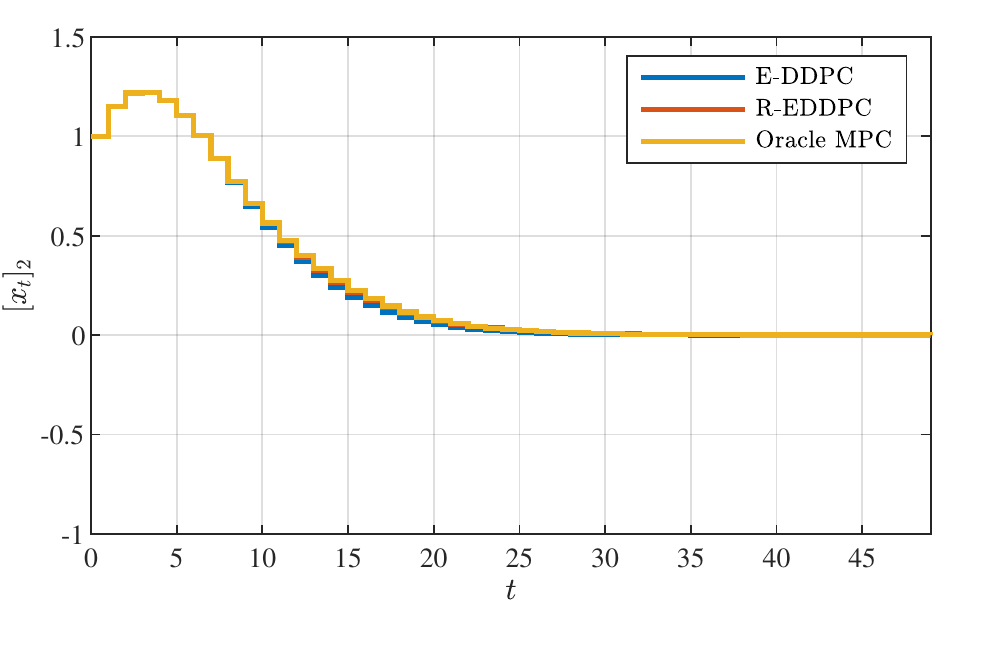}}
			\end{tabular}
		\end{center}
		\caption{SISO system with measurable state: comparison between the performance of the \emph{oracle} explicit MPC, E-DDPC and R-EDDPC over a noiseless test. The \emph{oracle} is obtained by using the true model of the system.}\label{Fig:ex1}
	\end{figure}
	Consider the system introduced in \cite[Section 7.1]{Bemporad2002b}, which is characterized by the following state dynamics:
	\begin{equation}\label{eq:model_ex1}
		x_{k+1}=\begin{bmatrix}
			0.7326 & -0.0861\\
			0.1722 & 0.9909
		\end{bmatrix}x_{k}+\begin{bmatrix}
			0.0609\\
			0.0064
		\end{bmatrix}u_{k},
		\end{equation}
	and assume that its states are measurable. To construct the matrices characterizing R-EDDPC, we have fed the plant with a random input sequence of length $N=100$, that is uniformly distributed in $[-5,5]$, so as to consider an experimental framework similar to the one introduced in \cite{Sassella2021}. The collected data are here corrupted by additive noise, \emph{i.e.,} 
	\begin{equation*}
		y_{k}^{d}=x_{k}^{d}+v_{k},
	\end{equation*}
	where $v \sim \mathcal{N}(0,\Upsilon)$, with $\Upsilon$ being a diagonal matrix chosen so as to yield an average Signal-to-Noise Ratio $\overline{\mbox{SNR}}$ over the two output channels equal to $20$~[dB].\\ 
	To compare R-EDDPC with E-DDPC, instead of solving \eqref{eq:regularized_MPC}, we explicitly solve a regularized version of the data-driven problem shown in \eqref{eq:relaxed_MPC} with weights chosen according to Theorem~\ref{thm:4}, namely
	\begin{subequations}\label{eq:regrelaxed_MPC}
		\begin{align}
			&\min_{\alpha,\bar{u},\bar{x}}~ \sum_{k=0}^{L-1} \left[\|\bar{x}_{k}\|_{\tilde{Q}}^{2}\!+\!\| \bar{u}_{k}\|_{\tilde{R}}^{2}\right]\!+\|\bar{x}_{L}\|_{\tilde{P}}^{2}+\rho_{\alpha}\|\alpha\|^{2}\label{eq:example_cost}\\
			&~~~\mbox{s.t. } \begin{bmatrix}
				\bar{u}_{[-n,L-1]}\\
				\bar{x}_{[-n,L-1]}
			\end{bmatrix}=\begin{bmatrix}
				H_{L+n}(u^{d})\\
				H_{L+n}(x^{d})
			\end{bmatrix}\alpha,\\
			&~~~\quad~~ \begin{bmatrix}
				\bar{u}_{[-n,-1]}\\
				\bar{x}_{[-n,-1]}	
			\end{bmatrix}=\chi_{0},\\
			&~~~\quad~~~ -2\leq \bar{u}_{k} \leq 2~~\forall k \in \mathcal{I}_{L},
		\end{align}
	\end{subequations} 
	with $L=2$, $\tilde{Q}=I_{2}$, $\tilde{R}=0.01$ and $\tilde{P}$ found by solving a data-driven Lyapunov function, as explained in \cite{DePersis2019}, and the input feasibility constraint corresponding to the one considered in \cite{Sassella2021}. We stress that the explicit law for this alternative MPC formulation can still be retrieved as in Section~\ref{Sec:derivation}, by properly augmenting $W_{d}^{\rho}$ and reshaping $\chi$ in \eqref{eq:regMPC_compact}. The regularization parameter $\rho_{\alpha}$ is instead tuned using cross-validation, i.e., by selecting the one minimizing the cost in \eqref{eq:example_cost} within a set of candidate values. Notice that such a procedure requires a closed-loop experiment for each value of $\rho_{\alpha}$ to be assessed. 
%	Not to bond the selected $\rho_{\alpha}$ to the data ultimately used to retrieve the actual explicit controller,
In this work, cross-validation is performed by considering a set of noisy state/input samples of length $N_{cv}=100$ gathered by feeding the plant with a new random input sequence uniformly distributed in $[-5,5]$. 
	This procedure leads to the choice of $\rho_{\alpha}=5$.\\ 
	Let us denote by \emph{oracle} explicit controller the law obtained by using the actual model of $\mathcal{P}$. \figurename{~\ref{Fig:ex1}} reports the comparison between the responses attained with R-EDDPC, E-DDPC and the \emph{oracle} explicit controller over a noiseless closed-loop test. Clearly, the difference between the three outcomes is generally negligible, with a slight discrepancy in the transient response that can be due to the different strategies employed to handle noise in R-EDDPC and E-DDPC. This results is in line with the expectations of Section~\ref{Sec:comparison}.\\   
	\begin{table}[!tb]
		\caption{SISO system with fully measurable state: $\overline{\mbox{SNR}}$ \emph{vs} $\rho_{\alpha}$ and $\mbox{RMSE}_{\mathcal{O}}$ (mean $\pm$ standard deviation) over $30$ Monte Carlo simulations for each noise level.}\label{tab:ex1}
		\begin{center}
			\begin{tabular}{ccc}
				$\overline{\mbox{SNR}}$~[dB] &$\rho_{\alpha}$ (mean$\pm$std) & RMSE$_{\mathcal{O}}$ (mean$\pm$std)\\
				\hline
				\hline
				40 & 0.7$\pm$0.2 & ~~~$(0.3\pm0.2)\cdot 10^{-2}$\\
				\hline
				30 & 2.0$\pm$0.6 & ~~~$(1.0\pm0.5)\cdot 10^{-2}$\\
				\hline
				20 & 6.2$\pm$2.2 &  ~~~$(2.7\pm2.0)\cdot 10^{-2}$\\
				\hline
				10 & 16.0$\pm$10.2 & ~~~$(9.8\pm5.8)\cdot 10^{-2}$\\
				\hline
			\end{tabular}
		\end{center}
	\end{table}
	We then assess the robustness of R-EDDPC to noisy data by considering $30$ different realizations of the datasets used in cross-validation and other $30$ realizations to construct the explicit law for increasing levels of noise, for a total of $60$ dataset for each noise level. We stress that a new hyper-parameter $\rho_{\alpha}$ is tuned for each of the $30$ training sets. To assess the performance of the retrieved R-EDDPC, we consider the following indicator
		\begin{equation}\label{eq:RMSE}
			\mbox{RMSE}_{\mathcal{O}}=\frac{1}{2}\sum_{i=1}^{2}\sqrt{\frac{1}{50}\sum_{t=0}^{49}([y_{t}]_{i}-[y_{t}^{\star}]_{i})},
		\end{equation}  
	which compares R-EDDPC with the \emph{oracle} explicit controller over the same noiseless test considered in \figurename{~\ref{Fig:ex1}}, with $\{y_{t}^{\star}\}_{t=0}^{N_{v}-1}$ being the output resulting from the use of the oracle controller. \tablename{~\ref{tab:ex1}} shows that both the sampled mean and standard deviation of the indicator in \eqref{eq:RMSE} remain small. Instead, the values of the regularization parameter obtained via cross-validation are modified to cope with the increasing noise level. This highlights that the hyper-parameter $\rho_{\alpha}$ can be actively exploited to improve the performance of the explicit predictive controller against noise. We remark that the indicators obtained when using the robustified approach presented in  Section~\ref{Sec:extensions} are comparable to the one in \tablename{~\ref{tab:ex1}}, when $\rho_{\sigma} \geq 100$.
	\subsection{Linearized four tank system}
	\begin{table}[!tb]
		\caption{Linearized four tank system: DD-PC \cite{Berberich2021} \emph{vs} R-EDDPC. Average time $\bar{\tau}$~[s] and worst case time $\tau_{wc}$~[s] needed to find the optimal control action and memory required for storage.}\label{tab:comparison_ex2}
		\begin{tabular}{cccc}
			& $\bar{\tau}$~[s] & $\tau_{wc}$~[s] & Memory~[kB]\\
			\hline
			DD-PC \cite{Berberich2021} & $1.9 \cdot 10^{-2}$ & $1.2 \cdot 10^{-1}$ & 3302\\
			\hline
			R-EDDPC & $8.5 \cdot 10^{-7}$ & $2.7 \cdot 10^{-5}$ &  2.1\\
			\hline
		\end{tabular}
	\end{table}
	Consider now the following fourth order system:
	\begin{subequations}\label{eq:model_ex2}
		\begin{align}
			& x_{k+1}=\!Ax_{k}+Bu_{k}+w_{k},\\
			&y_{k}=\begin{bmatrix}
				1 & 0 & 0 & 0\\
				0 & 1 & 0 & 0
			\end{bmatrix}x_{k}+v_{k},
		\end{align}
		where
		\begin{equation*}
			A=\begin{bmatrix}
				0.921 & 0 & 0.041 & 0\\
				0 & 0.918 & 0 & 0.033\\
				0 & 0 & 0.924 & 0\\
				0 & 0 & 0 & 0.937
			\end{bmatrix},~~~B=\begin{bmatrix}
				0.017 & 0.001\\
				0.001 & 0.023\\
				0 & 0.061\\
				0.072 & 0
			\end{bmatrix},
		\end{equation*}
		already considered in \cite{Berberich2021}. Differently from \cite{Berberich2021}, the state evolution is conditioned by the process noise $w_{k} \sim \mathcal{N}(0,\Delta)$, where $\Delta$ has been randomly chosen as
		\begin{equation*}
			\Delta=10^{-3} \begin{bmatrix}
				10 ~&~ 1 ~&~ 2 ~&~ 3\\
				1 ~&~ 10.01 ~&~ 2 ~&~ 1.5\\
				2 ~&~ 2 ~&~ 3 ~&~ 4\\
				3 ~&~ 1.5 ~&~ 4 ~&~ 7
			\end{bmatrix},
		\end{equation*}
		(notice that positive definiteness is verified), while the measurement are affected by an additive noise $v_{k} \sim \mathcal{N}(0,\Upsilon)$, with $\Upsilon$ chosen equal to $5.76 \cdot 10^{-4} I_{2}$, for the average output Signal-to-Noise ratio to be comparable to the one in \cite{Berberich2021}.  
	\end{subequations}
	Our goal is to force the inputs and outputs of the system to reach the equilibrium point
	\begin{equation*}
		u^{s} = \begin{bmatrix}
			1\\1\end{bmatrix}\!, ~~~ y^{s}=\begin{bmatrix} 0.65\\0.77\end{bmatrix}\!,
	\end{equation*}
	without requiring them to satisfy any value constraint over the prediction horizon. Since we share the same objective and specifications as \cite{Berberich2021}, we design the explicit predictive controller by considering the robust formulation in \eqref{eq:robustMPC} and by selecting the same parameters considered therein, \emph{i.e.,}
	\begin{equation*}
		L=30,~~Q=3I_{2},~~R=10^{-4}I_{2},~~\rho_{\alpha}=0.1,~~\rho_{\sigma}=10^{3}.
	\end{equation*}
	Analogously, we generate the data to construct the Hankel matrices in \eqref{eq:model} as in \cite{Berberich2021}, by feeding the plant with a random input sequence of length $N=400$, uniformly distributed within $[-1,1]$.\\
	In this setting, we compare the responses attained with the implicit DD-PC in \cite{Berberich2021} (denoted as $y_{t}^{i}$) and R-EDDPC, by considering the following indicator:
	\begin{equation*}
		\mbox{RMSE}_{I\!E}=\frac{1}{2} \sum_{j=1}^{2}\sqrt{\frac{1}{N_{v}}\sum_{t=0}^{N_v-1}([y_{t}]_{j}-[y_{t}^{i}]_{j})}
	\end{equation*}
	which allows us to assess the discrepancy in performance attained with the two controllers over a closed-loop test. Over a noiseless test of length $N_{v}=600$, we obtain $\mbox{RMSE}_{I\!E}=3.4 \cdot 10^{-7}$, which shows that the implicit and explicit law coincide in terms of performance, as expected. Nonetheless, in this case R-EDDPC is far more convenient from a computational perspective and in terms of memory occupation with respect to its implicit counterpart. Indeed, as shown in \tablename{~\ref{tab:comparison_ex2}}, the average and the worst case CPU times required to compute the optimal input by starting from $10^{4}$ randomly chosen values of $\chi_{0}$ in \eqref{eq:robust_init} are approximately $4$ orders of magnitude smaller for R-EDDPC, with the latter occupying only $6$\% of the memory required to store all the matrices needed for the implicit solution of \eqref{eq:robustMPC}.  While the first result is somehow expected, since with R-EDDPC the optimization problem in not solved in real-time, and the optimal input can be computed via simple function evaluations, the result on memory occupation is mainly due to the features of the considered MPC problem. Indeed, due to the lack of value constraints in the considered problem, R-EDDPC is a linear law. Therefore, this advantage in terms of memory requirements might be lost if value constraints are included in the problem. 
%	\subsection{Comparison with an explicit MPC law computed from an identified model}

	\begin{table}[!tb]
		\caption{Linearized four tank system: explicit MPC (E-MPC) \emph{vs} R-EDDPC. Comparison over $30$ Monte Carlo runs with respect to the percentage of unstable instances in closed-loop and the values of $\mathcal{J}$ (mean $\pm$ standard deviation) in \eqref{eq:KPI}.}\label{tab:comparison2_ex2}
		\begin{tabular}{ccc}
			& Unstable runs [\%] & $\mathcal{J}$ (mean$\pm$std)\\
			\hline
			\hline
			E-MPC+N4SID & 83 \%& 81.94$\pm$122.18\\
			\hline
			R-EDDPC & 0 \%& 9.00$\pm$0.04\\
			\hline
		\end{tabular}
	\end{table}
	\begin{figure*}[!tb]
		\begin{center}
			\begin{tabular}{c}
				\subfigure[First output]{\begin{tabular}{cc}
						\includegraphics[scale=.7,trim=1cm 1cm 8cm 22.5cm,clip]{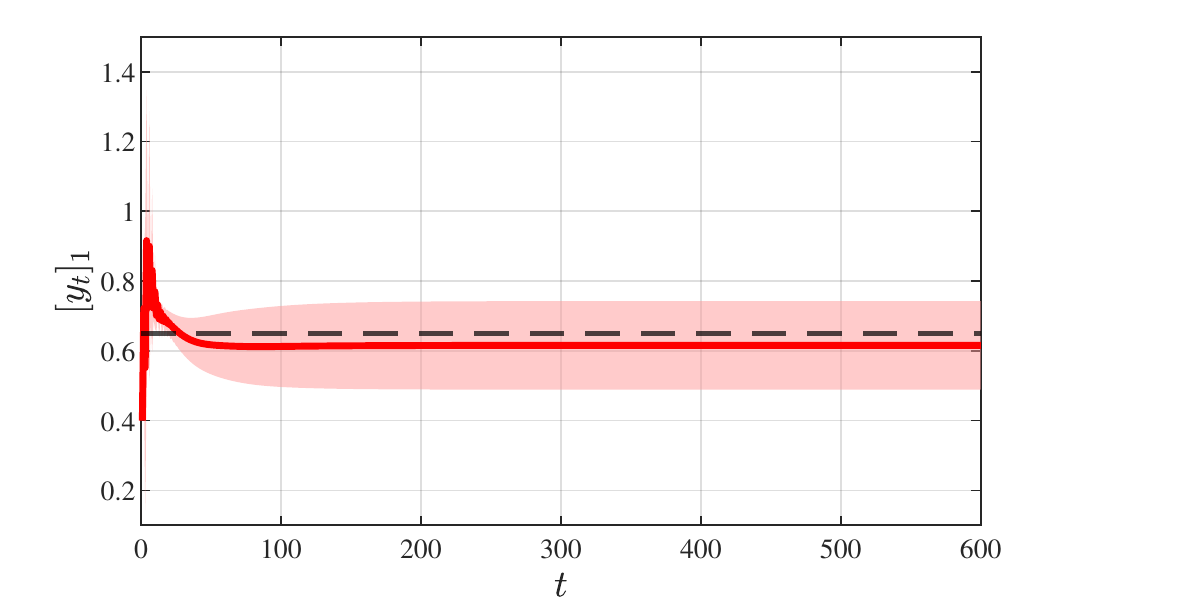} & \includegraphics[scale=.7,trim=1cm 1cm 8cm 22.5cm,clip]{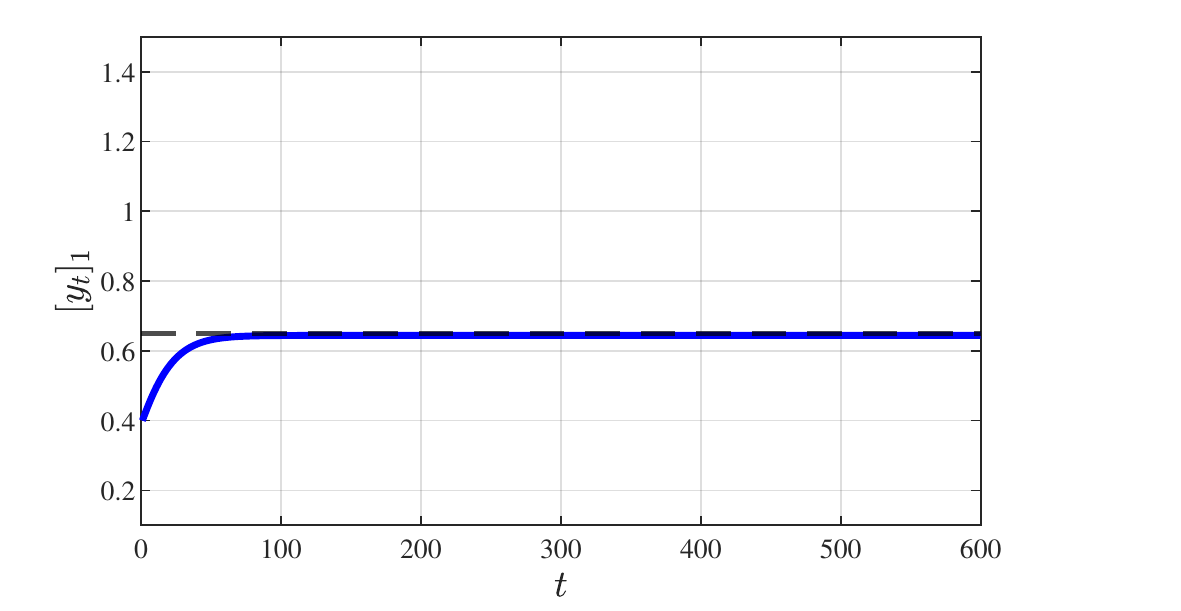}
				\end{tabular}}\\
				\subfigure[Second output]{\begin{tabular}{cc}
						\includegraphics[scale=.7,trim=1cm 1cm 8cm 22.5cm,clip]{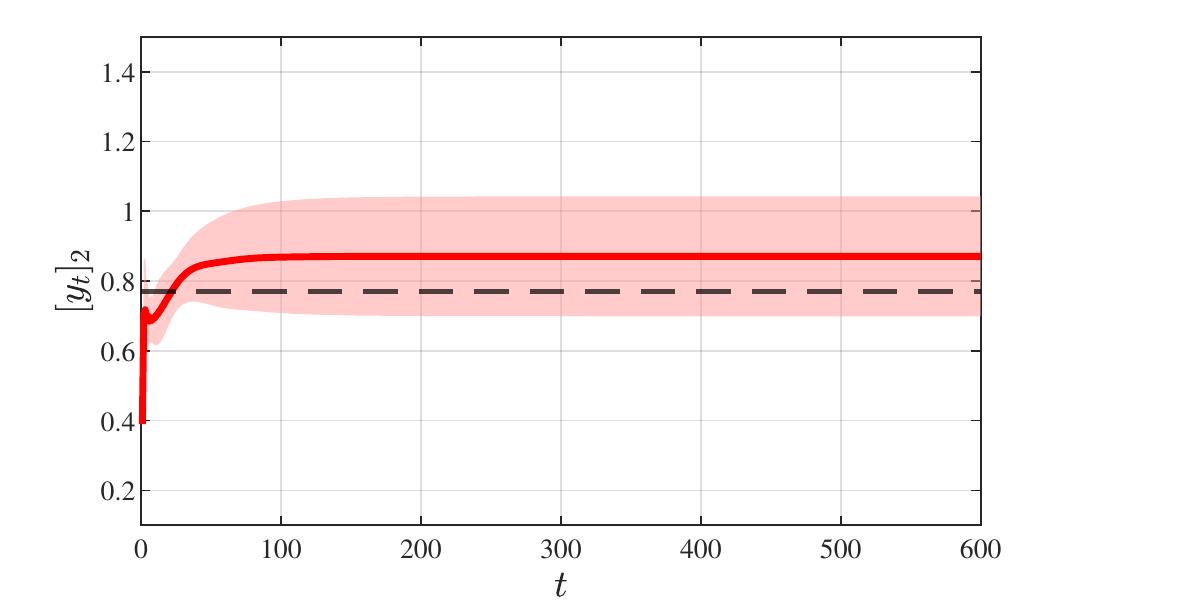} & \includegraphics[scale=.7,trim=1cm 1cm 8cm 22.5cm,clip]{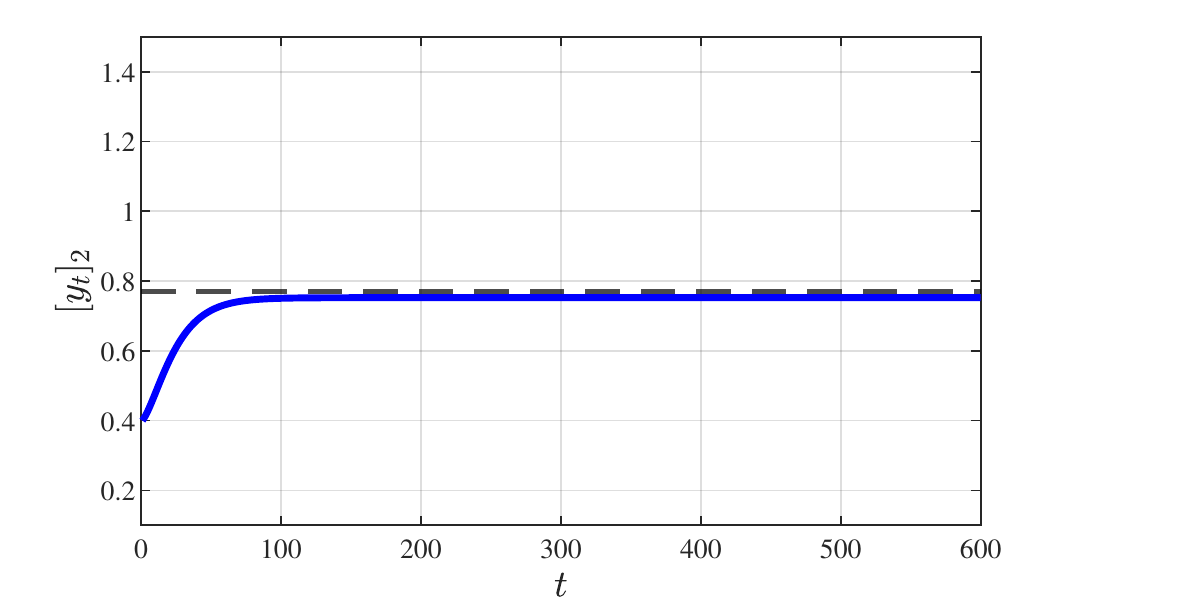}
				\end{tabular}}
			\end{tabular}
		\end{center}
		\caption{Linearized four tank system: E-MPC with identified model (left panels) \emph{vs} R-EDDPC (right panels). Comparison between the mean (line) and standard deviations (shaded areas) of the outputs resulting by the use of the two controllers, limited to stable closed-loop instances only. When looking at the performance of R-EDDPC, the standard deviation is negligible.} \label{Fig:comparison_ex2}
	\end{figure*}
	We now compare the performance attained by R-EDDPC with the ones achieved by designing an explicit predictive controller via the conventional two-phase strategy. By using the same data exploited to construct the Hankel matrices for R-EDDPC, we thus identify a model for the system in \eqref{eq:model_ex2} via the N4SID approach \cite{VanOverschee1994,MatlabSYSID} and, then, we design a model-based predictive controller as in \cite{Bemporad2002b}. This two-stage solution is here denoted as E-MPC. Note that, since the latter needs the state of the system for the computation of the optimal input, E-MPC further requires the design of a Kalman filter \cite{Kalman1960} based on the identified model. In this work, this additional step is carried out under the assumption that an oracle provides us with the true covariance matrices of the process and measurement noise, so as to skip the Kalman filter design phase. We point out that the need of the state estimator already highlights an intrinsic advantage of R-EDDPC, which relies on input/output data only and, thus, solely involve the construction of the Hankel matrices to be designed and deployed.\\
	Initially, we evaluate the robustness of R-EDDPC and E-MPC to noisy data. To this end, we consider $30$ datasets, characterized by different realizations of the process and measurement noise in \eqref{eq:model_ex2}. For each of them, we identify a model of the system and then design the explicit MPC and the R-EDDPC laws. For each of the explicit controllers obtained with E-MPC and R-EDDPC, we then run the same noiseless test and evaluate the attained closed-loop performance through the following index:         
	\begin{equation}\label{eq:KPI}
		\mathcal{J}=\sum_{t=0}^{N_{v}-1} \left[\|y_{t}-y^{s}\|_{Q}^{2}+\|u_{t}-u^{s}\|_{R}^{2} \right],
	\end{equation}    
	where $N_{v}=600$. \tablename{~\ref{tab:comparison2_ex2}} reports its mean value and standard deviation over stable closed-loop instances only. Clearly, R-EEDPC outperforms the standard model-based approach in terms of robustness with respect to different realizations of the training set, generally leading to better performance, as also confirmed by the results reported in \figurename{~\ref{Fig:comparison_ex2}}. This remarkable difference between the two approaches can be linked to the poor quality of the identified model, which indeed results in an average fit of $15$ \% in validation, thus jeopardizing the performance of both E-MPC and the Kalman filter. On the other hand, avoiding a preliminary identification step, adding the regularization term and further robustifying R-EDDPC via the slack variable have proven to be a valid strategy to handle both measurement and process noise. We point out that similar results are also obtained by choosing alternative covariance matrices for the noise in \eqref{eq:model_ex2} at random, while maintaining a similar Signal-to-Noise Ratio on the output.
		\begin{remark}[Case-study dependend conclusions]
		We wish to stress that the tests presented here, but performed with no process noise, have led to much better performance of the identification procedure and, thus, of the model-based controller. It follows that this case study must be considered as a \emph{special} case, nonetheless highlighting that it \emph{might} be worthwhile to map the data directly onto the predictive controller, instead of first identifying the model of the system. Further research is % indeed 
		needed to generalize such a statement.
	\end{remark}
	\section{Conclusions}\label{Sec:conclusions}
	In this paper, we have presented a data-driven regularized explicit predictive controller (R-EDDPC), which can be designed with a batch of properly generated input/output data. The effectiveness of R-EDDPC has been proven on two benchmark simulation examples, showing its correspondence with E-DDPC in \cite{Sassella2021}. The numerical results additionally show that R-EDDPC might outperform an explicit MPC relying on an poorly identified model of the system to be controlled.\\
	Due to the crucial role of the regularization parameter on the performance attained in closed-loop, future works will be devoted to devise hyper-parameter tuning techniques not involving closed-loop experiments. In addition, future research will explore strategies to extend the latter to handle nonlinear systems.  	

	\bibliographystyle{plain}
	\bibliography{R-EDDPC}
	
	\appendix
	\section{Appendix A}
	\subsection{Proof of Lemma~\ref{lemma:1}}\label{appendix:A}	
	\begin{itemize}
		\item[$(i)$] Assume that the state is \emph{fully measurable} and consider the predictive model in \eqref{eq:EDDPC_model}. By propagating it over time, we can characterize the stack of predicted states $\bar{x}_{[1,L]}$ as follows:
		\begin{equation*}
			\bar{x}_{[1,L]}=\begin{bmatrix}
				\Gamma_{d} ~&~ \Xi_{d}
			\end{bmatrix}\begin{bmatrix}
				\bar{u}_{[0,L-1]}\\
				\bar{x}_{0}
			\end{bmatrix},
		\end{equation*}
		where $\Gamma_{d}$ and $\Xi_{d}$ are defined as in \cite[Appendix A]{Sassella2021}. Since $u_{d}$ is persistently exciting of order $L+2n$, it exists $\alpha \in \mathbb{R}^{n_{\alpha}}$ such that 
		\begin{equation*}
			\bar{x}_{[1,L]}=\begin{bmatrix}
				\Gamma_{d} ~&~ \Xi_{d}
			\end{bmatrix}\begin{bmatrix}
				\mathcal{H}_{u}^{F}\\
				\mathcal{H}_{x}^{F}
			\end{bmatrix}\alpha=H_{L}(x_{d}^{+})\alpha,
		\end{equation*} 
		where $x_{d}^{+}=\{x_{k}^{d}\}_{k=n+1}^{N+n}$, $\mathcal{H}_{u}^{F}$ and $\mathcal{H}_{x}^{F}$ are defined as in \eqref{eq:aid_matrices} and the last equality stems from the definitions of $\Gamma_{d}$ and $\Xi_{d}$.
		
		Consider now the model in \eqref{eq:model} and decouple the initial state from the predicted ones, by relying on the decomposition in \eqref{eq:aid_matrices}. Let us introduce the transformation $\mathcal{T} \in \mathbb{R}^{Ln \times L(m+p)}$ such that:
		\begin{equation*}
			\bar{x}_{[1,L]}=\mathcal{T}\begin{bmatrix}
				\bar{u}_{[0,L-1]}\\
				\bar{x}_{[0,L-1]}
			\end{bmatrix}.
		\end{equation*} 
		By premultiplying both sides of \eqref{eq:model} for $\mathcal{T}$, we obtain
		\begin{equation*}
			\bar{x}_{[1,L]}=\mathcal{T}\begin{bmatrix}
				\mathcal{H}_{u}^{F}\\
				\mathcal{H}_{x}^{F}
			\end{bmatrix}\alpha=H_{L}(x_{d}^{+})\alpha,
		\end{equation*}
		where we have exploited the properties of $\mathcal{T}$ and we have replaced the measured output with the state, since they coincide. Therefore,  the predictive models in \eqref{eq:model} and \eqref{eq:EDDPC_model} are equivalent.
		\item[$(ii)$] Assume that the state is \emph{not fully measurable}. Let $\tilde{\mathcal{T}}$ be the transformation that allows to reconstruct the predicted state sequence $\bar{z}_{[1,L]}$, with $\bar{z}_{k}$ defined as in \eqref{eq:nonminimal} for $k=1,\ldots,L$, \emph{i.e.,}
		\begin{equation*}
			\bar{z}_{[1,L]}=\tilde{\mathcal{T}}\begin{bmatrix}
				\bar{u}_{[-n,L-1]}\\
				\bar{y}_{[-n,L-1]}
			\end{bmatrix}.
		\end{equation*} 
		By premultiplying both sides of \eqref{eq:model} for this transformation matrix, we obtain:
		\begin{equation*}
			\bar{z}_{[1,L]}=\tilde{\mathcal{T}}\begin{bmatrix}
				H_{L+n}(u_{d})\\
				H_{L+n}(y_{d})
			\end{bmatrix}\alpha=H_{L}(z_{d}^{+})\alpha,
		\end{equation*}
		where $z_{d}^{+}=\{z_{k}^{d}\}_{k=n+1}^{N+n}$ and 
		\begin{equation*}
			z_{k}^{d}=\begin{bmatrix}
				(u_{k-n}^{d})' & \cdots & (u_{k-1}^{d})' & (y_{k-n}^{d})' & \cdots & (y_{k-1}^{d})' 
			\end{bmatrix}'.
		\end{equation*}
		Let us now recast the model in \eqref{eq:EDDPC_model} according to its input/ouput counterpart (see \cite[Section VI]{DePersis2019}), \emph{i.e.,}
		\begin{equation}
			\bar{z}_{k+1}=Z_{1,N}\begin{bmatrix}
				U_{0,1,N}\\
				Z_{0,N}
			\end{bmatrix}^{\dagger}\begin{bmatrix}
				\bar{u}_{k}\\
				\bar{z}_{k}
			\end{bmatrix},
		\end{equation}
		where
		\begin{align}
			&Z_{0,N}=\begin{bmatrix}
				x_{n}^{d} & x_{n+1}^{d} &\cdots & x_{N+n-1}^{d}
			\end{bmatrix},\label{eq:Z0}\\
			& Z_{1,N}=\begin{bmatrix}
				x_{n+1}^{d} & x_{n+2}^{d} & \cdots & x_{N+n}^{d}
			\end{bmatrix}.\label{eq:Z1}\\
		\end{align}
		By propagating this model over time, we can express the predictive state sequence as
		\begin{equation*}
			\bar{z}_{1,L}=\begin{bmatrix}
				\tilde{\Gamma}_{d} ~&~ \tilde{\Xi}_{d}
			\end{bmatrix}\begin{bmatrix}
				\bar{u}_{[0,L-1]}\\
				\bar{z}_{0}
			\end{bmatrix}=\begin{bmatrix}
				\tilde{\Gamma}_{d} ~&~ \tilde{\Xi}_{d}
			\end{bmatrix}\begin{bmatrix}
				\mathcal{H}_{u}^{F}\\
				\mathcal{H}_{y}^{F}
			\end{bmatrix}\alpha,
		\end{equation*}
		where we have exploited the definitions of $\mathcal{H}_{u}^{F}$ and $\mathcal{H}_{y}^{F}$ in \eqref{eq:aid_matrices} and the fact that the input sequence is persistently exciting of order $L+2n$. Since $\tilde{\Gamma}_{d}$ and $\tilde{\Xi}_{d}$ embed the data-driven counterpart of the (unknown) model of $\mathcal{P}$, it can be straightforwardly proven that the following holds:
		\begin{equation*}
			\bar{z}_{1,L}=\begin{bmatrix}
				\tilde{\Gamma}_{d} ~&~ \tilde{\Xi}_{d}
			\end{bmatrix}\begin{bmatrix}
				\mathcal{H}_{u}^{F}\\
				\mathcal{H}_{y}^{F}
			\end{bmatrix}\alpha=H_{L}(z_{d}^{+})\alpha,
		\end{equation*}
		where $z_{d}^{+}$ is defined as before, thus concluding the proof.
	\end{itemize}
	
	\subsection{Proof of Lemma~\ref{lemma:2}}\label{appendix:B}
	Consider the initial conditions in \eqref{eq:initial_constraint} and \eqref{eq:EDDPC_init}. When the state is fully measured, we can reconstruct the initial state from input/state data as
	\begin{equation*}
			\bar{x}_{0}=T\begin{bmatrix}
				\bar{u}_{[-n,-1]}\\
				\bar{x}_{[-n,-1]}
			\end{bmatrix},
	\end{equation*}
	 where $T \in \mathbb{R}^{n \times (m+p)n}$ is the matrix characterizing this coordinate transformation.	By premultiplying both sides of \eqref{eq:initial_constraint} by $T$, it can be straightforwardly proven that the following holds:
	\begin{equation*}
		\bar{x}_{0}=T\chi_{0}=x,
	\end{equation*}
	which corresponds to the initial condition imposed via \eqref{eq:EDDPC_init}. Instead, when the state is not measured, we can rely on the use of the nonminimal representation \eqref{eq:nonminimal} to reformulate \eqref{eq:EDDPC_init} as:
	\begin{equation*}
		\bar{z}_{0}=\bar{z}.
	\end{equation*}
	Based on the definition of $\bar{z}_{k}$ in \eqref{eq:nonminimal}, it easily follows that \eqref{eq:EDDPC_init} is equal to \eqref{eq:initial_constraint}.\\ 
	As for the conditions that have to be satisfied for the feasibility constraints to hold, they stem straightforwardly from the definition of $\bar{y}_{k}$ when the state if fully measurable, and the one of $\bar{z}_{k}$ when $y_{k} \neq x_{k}$, thus concluding the proof. 	
	\subsection{Proof of Theorem~\ref{thm:4}}\label{appendix:C}
	The equivalence between the constraints of the problems in \eqref{eq:EDDPC_prob} and \eqref{eq:relaxed_MPC} follows from the results in Lemma \ref{lemma:1}-\ref{lemma:2}. The conditions on the weighting matrices can instead be proven as follows.
	\begin{itemize}
		\item[$(i)$] When the state is fully measured, the transformation matrix $T$ allows one to reconstruct the terminal cost characterizing problem~\eqref{eq:EDDPC_cost}. Moreover, since the state is measured, we can replace $\bar{y}_{k}$ with $\bar{x}_{k}$. Therefore, the equivalence of \eqref{eq:relaxed_MPC} and \eqref{eq:EDDPC_prob} can be straightforwardly verified by choosing the weighting matrices $Q$, $R$ and $P$ as indicated in the statement.  
		\item[$(ii)$] When the state is not fully measured, according to the nonminimal representation in \eqref{eq:nonminimal}, the problem in \eqref{eq:EDDPC_cost} has to be modified as follows:
		\begin{subequations}\label{eq:modified_EDDPC}
			\begin{align}
				& \min_{\bar{u}_{[0,L-1]}} \sum_{k=0}^{L-1} \left[\|\bar{z}_{k}\|_{\tilde{Q}}^{2}+\|\bar{u}_{k} \|_{\tilde{R}}^{2}\right]+\|\bar{z}_{L}\|_{\tilde{P}}^{2} \label{eq:EDDPC_modcost}\\
				&~~\mbox{s.t. }~~\bar{z}_{k+1}\!=\!Z_{1,N}\!\begin{bmatrix}
					U_{0,1,N}\\
					\hline
					Z_{0,N}
				\end{bmatrix}^{\!\dagger} \!\begin{bmatrix}
					\bar{u}_{k}\\
					\bar{z}_{k}
				\end{bmatrix}\!, ~~k\!=\!0,\ldots,L\!-\!1,\\
				& \bar{z}_0 = \bar{z},\\
				& \bar{u}_{k} \in \mathbb{U}, ~~ \bar{z}_{k} \in \mathbb{X}, 
			\end{align}
		\end{subequations}
		where $Z_{0,N}$ and $Z_{1,N}$ are defined as in \eqref{eq:Z0} and \eqref{eq:Z1}, respectively. By substituting the matrices in \eqref{eq:weights} into \eqref{eq:modified_EDDPC}, it can be easily seen that the cost $\tilde{J}(\bar{u}_{[0,L-1]})$ in \eqref{eq:EDDPC_modcost} is equivalent to:
		\begin{align*}
			&\tilde{J}(\bar{u}_{[0,L-1]})=\|\bar{u}_{-1}\|_{R}^{2}+\|\bar{y}_{-1}\|_{Q}^{2}\\
			&\qquad +\sum_{k=0}^{L-1} \|\bar{y}_{k}\|_{Q}^{2}+\|\bar{u}_{k}\|_{R}^{2}+\bigg\|\begin{bmatrix}
				\bar{u}_{[L-n,L-1]}\\
				\bar{y}_{[L-n,L-1]}
			\end{bmatrix} \bigg\|_{P}^{2}.
		\end{align*}
		Since in \eqref{eq:modified_EDDPC} we optimize over $\bar{u}_{[0,L-1]}$, the two initial terms in the cost can be neglected, thus concluding the proof.
	\end{itemize}

\end{document}